\newif\ifArxiv
\Arxivtrue

\ifArxiv
\documentclass[11pt]{article}
\usepackage[a4paper, margin=25mm]{geometry}
\else
\documentclass[a4paper,
cleveref, autoref
]{socg-lipics-v2021}
\fi

\usepackage{graphicx}
\usepackage{color}
\usepackage{amsmath,amsthm}
\usepackage{amssymb}
\usepackage{xspace}
\usepackage{epsfig}
\usepackage{xfrac}
\usepackage[dvipsnames]{xcolor}
\usepackage[colorinlistoftodos,prependcaption,textsize=tiny]{todonotes}

\newcommand {\mm}[1]   {\ifmmode{#1}\else{\mbox{\(#1\)}}\fi}

\newcommand{\ignore}[1]{}

\newsavebox{\smallProofsym}                            
\savebox{\smallProofsym}                               %
{
\begin{picture}(6,6)                                   %
\put(0,0){\framebox(6,6){}}                            %
\put(0,2){\framebox(4,4){}}                            %
\end{picture}                                          %
}                                                      %
\makeatletter
\long\def\@makecaption#1#2{%
  \vskip\abovecaptionskip
  \sbox\@tempboxa{\small #1: #2}%
  \ifdim \wd\@tempboxa >\hsize
    \small #1: #2\par
  \else
    \global \@minipagefalse
    \hb@xt@\hsize{\hfil\box\@tempboxa\hfil}%
  \fi
  \vskip\belowcaptionskip}
\makeatother

\ifArxiv
 \newtheorem{lemma}{Lemma}
 \newtheorem{claim}{Claim}

 \newtheorem{theorem}{Theorem}

\fi

\newtheorem*{namedthm}{\namedthmname}
\newcounter{namedthm}
\makeatletter
\newenvironment{named}[1]
  {\def\namedthmname{#1}%
   \refstepcounter{namedthm}%
   \namedthm\def\@currentlabel{#1}}
  {\endnamedthm}
\makeatother

\newcommand{\InfSup}        {\mm{c_0}}
\newcommand{\SupSup}        {\mm{C_0}}
\newcommand{\Rspace}        {\mm{{\mathbb R}}}
\newcommand{\Zspace}        {\mm{{\mathbb Z}}}
\newcommand{\Hdisk}         {\mm{{\mathbb H}}}

\newcommand{\FCC}           {\mm{{\rm FCC}}}
\newcommand{\Hgroup}[2]     {\mm{H_{#1}{({#2})}}}
\newcommand{\Dgm}[2]        {\mm{\rm Dgm}_{#1}{({#2})}}
\newcommand{\Domain}[2]     {\mm{\rm Dom}_{#1}{({#2})}}

\newcommand{\Image}[2]      {\mm{\rm Im}_{#1}{({#2})}}
\newcommand{\Kernel}[2]     {\mm{\rm Ker}_{#1}{({#2})}}
\newcommand{\cutoff}        {\mm{\omega_0}}

\newcommand{\Length}[1]     {\mm{|{#1}|}}
\newcommand{\MST}[1]        {\mm{{\rm MST}{({#1})}}}

\newcommand{\hexMST}[1]     {\mm{{\rm MST_{hex}}{({#1})}}}
\newcommand{\MSTratio}[2]   {\mm{\mu{({#1},{#2})}}}
\newcommand{\AF}            {\mm{A}}
\newcommand{\AL}            {\mm{\Lambda}}
\newcommand{\BL}            {\mm{\Delta}}

\newcommand{\Frontier}[2]   {\mm{\partial D_{#1}{({#2})}}}
\newcommand{\Tree}[2]       {\mm{T_{#1}}}
\newcommand{\Uree}[2]       {\mm{U_{#1}}}
\newcommand{\Vree}[2]       {\mm{V_{#1}}}
\newcommand{\Thicken}[2]    {\mm{D_{#1}{({#2})}}}

\newcommand{\card}[1]       {\mm{{\#}{#1}}}

\newcommand{\Edist}[2]      {\mm{\|{#1}-{#2}\|}}

\newcommand{\hexdist}[2]    {\mm{\|{#1}-{#2}\|_{hex}}}
\newcommand{\norm}[1]       {\mm{\|{#1}\|}}
\newcommand{\perimeter}[1]  {\mm{\rm per}{({#1})}}

\newcommand{\aaa}           {\mm{\bf a}}
\newcommand{\bbb}           {\mm{\bf b}}
\newcommand{\ccc}           {\mm{\bf c}}
\newcommand{\uuu}           {\mm{\bf u}}
\newcommand{\vvv}           {\mm{\bf v}}

\newcommand{\xxx}           {\mm{\bf x}}
\newcommand{\yyy}           {\mm{\bf y}}
\newcommand{\zzz}           {\mm{\bf z}}
\newcommand{\ee}            {\mm{\varepsilon}}
\newcommand{\dd}            {\mm{\delta}}

\newcommand{\Skip}[1]       {}

\definecolor{blue-green}{rgb}{0.0, 0.5, .8}

\title{The Euclidean MST-ratio of Bi-colored Lattices}

\ifArxiv
\author{S.\ Cultrera di Montesano, O.\ Draganov, H.\ Edelsbrunner and M.\ Saghafian}
\else
\titlerunning{The Euclidean MST-ratio of Bi-colored Lattices}

\title{The Euclidean MST-ratio for Bi-colored Lattices}

\titlerunning{The Euclidean MST-ratio of Bi-colored Lattices}

\author{Sebastiano Cultrera di Montesano}{ISTA (Institute of Science and Technology Austria), Kloster\-neu\-burg, Austria}{Sebastiano.Cultrera@ist.ac.at}{https://orcid.org/0000-0001-6249-0832}{}

\author{Ond\v{r}ej Draganov}{ISTA (Institute of Science and Technology Austria), Kloster\-neu\-burg, Austria}{Ondrej.Draganov@ist.ac.at}{https://orcid.org/
0000-0003-0464-3823}{}

\author{Herbert Edelsbrunner}{ISTA (Institute of Science and Technology Austria), Kloster\-neu\-burg, Austria}{Herbert.Edelsbrunner@ist.ac.at}{https://orcid.org/0000-0002-9823-6833}{}

\author{Morteza Saghafian}{ISTA (Institute of Science and Technology Austria), Kloster\-neu\-burg, Austria}{Morteza.Saghafian@ist.ac.at}{https://orcid.org/0000-0002-4201-5775}{}

\authorrunning{Cultrera di Montesano, Draganov, Edelsbrunner, and Saghafian}

\Copyright{Sebastiano Cultrera di Montesano, Ond\v{r}ej Draganov, Herbert Edelsbrunner, and Morteza Saghafian}


\funding{\footnotesize
  This project has received funding from the  European Research Council (ERC) under the European   Union's Horizon 2020 research and innovation  programme, grant no.\ 788183, from the Wittgenstein Prize, Austrian Science Fund (FWF), grant no.\ Z 342-N31, and from the DFG Collaborative Research Center TRR 109, `Discretization in Geometry and Dynamics', Austrian Science Fund (FWF), grant no.\ I 02979-N35.}

\EventEditors{}
\EventNoEds{0}
\EventLongTitle{}
\EventShortTitle{}
\EventAcronym{}
\EventYear{}
\EventDate{}
\EventLocation{}
\EventLogo{}
\SeriesVolume{}
\ArticleNo{}
\fi

\begin{document}
\maketitle

\begin{abstract}
  Given a finite set, $\AF \subseteq \Rspace^2$, and a subset, $B \subseteq \AF$, the \emph{MST-ratio} is the combined length of the minimum spanning trees of $B$ and $\AF \setminus B$ divided by the length of the minimum spanning tree of $\AF$.
  The question of the supremum, over all sets $\AF$, of the maximum, over all subsets $B$, is related to the Steiner ratio, and we prove this sup-max is between $2.154$ and $2.427$.
  Restricting ourselves to $2$-dimensional lattices, we prove that the sup-max is $2$, while the inf-max is $1.25$.
  By some margin the most difficult of these results is the upper bound for the inf-max, which we prove by showing that the hexagonal lattice cannot have MST-ratio larger than $1.25$.
\end{abstract}


\section{Introduction}
\label{sec:1}

The recent development of measuring the interaction between two or more sets of points with methods from topological data analysis motivates the discrete geometric question about minimum spanning trees studied in this paper; see \cite{Car09,EdHa10} for background in this general area.
We refer to the measured interaction as \emph{mingling},
in which higher values corresponding to more mingling.
The ambiguity of the term is deliberate and leaves the concrete meaning to the geometric and algebraic constructions described in \cite{CDES23}.
As explained in the appendix of the current paper, one of these measurements can be expressed in elementary terms:
\begin{named}{Definition}
  Given a finite set, $\AF \subseteq \Rspace^2$, we write $\MST{\AF}$ for the \emph{(Euclidean) minimum spanning tree} of the complete graph on $\AF$, with edge weights equal to the distances between the points.
  For $B \subseteq \AF$, the \emph{MST-ratio} of $\AF$ and $B$ is the combined length of the minimum spanning trees of $B$ and $\AF \setminus B$, divided by the length of the minimum spanning tree of $\AF$:
  \begin{align}
    \MSTratio{\AF}{B}  &= \frac{\Length{\MST{B}}+\Length{\MST{\AF \setminus B}}}{\Length{\MST{\AF}}} .
  \end{align}
\end{named}
To make use of this measure for statistical or other purposes, we ought to know how small and how large the ratio can get (the extremal question), and how it behaves for random data.
A first result in the latter direction can be found in \cite{DPT23}, who prove that for points chosen uniformly at random in the unit square, the expected MST-ratio for a random partition into two subsets is at least $\sqrt{2} - \ee$, for any $\ee > 0$.
In the non-random setting, we study the maximum MST-ratio, over all partitions of $A$ into two sets, and consider both the infimum and supremum of the maximum, over all sets in a class of point sets.
If these sets are infinite, like for example $2$-dimensional lattices, then we talk about the supremum rather than the maximum MST-ratio.

\smallskip
Given any set, $\AF$, the minimum MST-ratio is achieved by removing the longest edge from $\MST{\AF}$ and letting $B$ and $\AF \setminus B$ be the vertices of the resulting two trees, so it is less than $1$.
Indeed, any other partition of $\AF$ would produce two minimum spanning trees that together are at least as long as $\MST{B}$ and $\MST{\AF \setminus B}$.
More interestingly, the maximum MST-ratio is related to the \emph{Steiner ratio} of the Euclidean plane \cite{GiPo68,JaKo34}, and we exploit this connection to prove that the supremum is between $2.154$ and $2.427$
(Theorem~\ref{thm:maximum_MST-ratio_for_finite_sets} in Section~\ref{sec:2}).
The infimum of the maximum is again less interesting: allowing ourselves to pick points arbitrarily close to each other,
and one far away, this infimum can be seen to be arbitrarily close to $1$.

\smallskip
This motivates us to study the MST-ratio for a restricted class of sets, and our choice are the (Euclidean) lattices, which are well studied objects with many applications in mathematics and beyond; see e.g.\ \cite{Zhi15}.
Since we optimize over subsets of an infinite set, we talk about the supremum rather than the maximum, and taking a sequence of progressively larger but finite portions of such a lattice, we have well defined minimum spanning trees and can study the asymptotic behavior of the MST-ratio.
Our main result is that the supremum MST-ratio of the hexagonal lattice is $1.25$ (Theorem~\ref{thm:maximum_MST-ratio_of_hexagonal_lattice} in Section~\ref{sec:4}).
Observe that this is an upper bound on the infimum, over all lattices, of the supremum MST-ratio.
We complement this with a matching lower bound (Claim~\ref{clm:lower_bound_for_infmax} in Section~\ref{sec:3}), and with matching lower and upper bounds for the supremum, over all lattices, of the supremum MST-ratio, which we establish is $2$ (Claims~\ref{clm:lower_bound_for_supmax} and \ref{clm:upper_bound_for_supmax} in Section~\ref{sec:3}).

\section{The Maximum MST-ratio for Finite Sets}
\label{sec:2}

The main question we ask is to what extent two minimum spanning trees can be longer than a single minimum spanning tree of the same points; see the definition of the MST-ratio of a set $\AF \subseteq \Rspace^2$ and a subset $B \subseteq A$ in the introduction.
We are interested in the maximum MST-ratio, over all subsets $B \subseteq A$, and in the supremum and infimum of this maximum, over all finite sets $\AF \subseteq \Rspace^2$.

\smallskip
The supremum is related to the well-studied Steinter tree problem.
Given a finite set, $X \subseteq \Rspace^2$, the \emph{Steiner tree} of $X$ is the minimum spanning tree of $X \cup B$, in which $B = B(X)$ is chosen to minimize the length of this tree.
The \emph{Steiner ratio} of the Euclidean plane is the infimum of the length ratio, $\Length{\MST{X \cup B}} / \Length{\MST{X}}$, over all finite sets $X$ and $B$ in the plane.
There are sets $X \subseteq \Rspace^2$ for which the ratio is only $\sqrt{3}/2 = 0.866\ldots$;
take for example the vertices of an equilateral triangle as $X$ and the barycenter of this triangle as the sole point in $B$.
It is conjectured that $\sqrt{3}/2$ is the Steiner ratio of the Euclidean plane \cite{GiPo68}, but the current best lower bound proved in \cite{ChGr85} is only $0.824\ldots$.
We use this bound to prove upper and lower bounds for the supremum maximum MST-ratio:
\begin{theorem}
  \label{thm:maximum_MST-ratio_for_finite_sets}
  The supremum, over all finite $\AF \subseteq \Rspace^2$, of the maximum, over all subsets $B \subseteq \AF$, of the MST-ratio satisfies $2.154 \leq \sup_{\AF} \max_B \MSTratio{\AF}{B} \leq 2.427$.
\end{theorem}
\begin{proof}
  We first prove the upper bound.
  Since $B$ is a subset of $\AF$, the MST of $\AF$ cannot be shorter than the Steiner tree of $B$.
  Similarly, the MST of $\AF$ cannot be shorter than the Steiner tree of $\AF \setminus B$.
  Hence, $\Length{\MST{\AF}} \geq 0.824\ldots \cdot \Length{\MST{B}}$ and $\Length{\MST{\AF}} \geq 0.824\ldots \cdot \Length{\MST{\AF \setminus B}}$.
  It follows that the ratio satisfies
  \begin{align}
    \MSTratio{\AF}{B}  &\leq  \frac{2 \cdot [\Length{\MST{B}} + \Length{\MST{\AF \setminus B}}]}{0.824\ldots \cdot [\Length{\MST{B}} + \Length{\MST{\AF \setminus B}}]}
      =  2.426\ldots .
  \end{align}
  This inequality holds for every $B \subseteq \AF$.
  We second prove the lower bound for the sup-max by constructing a set $\AF$ of seven points that implies the inequality.
  Let $B \subseteq \AF$ be the three vertices of an equilateral triangle with unit length edges, and let $\AF \setminus B$ be the vertices of another equilateral triangle with unit length edges, but this time together with the barycenter.
  Hence, $\Length{\MST{B}} = 2$ and $\Length{\MST{\AF \setminus B}} = \sqrt{3}$.
  Assuming the distance between corresponding vertices of the two equilateral triangles is less than $\ee > 0$, we have $\Length{\MST{\AF}} < \sqrt{3} + 3\ee$.
  This implies 
  \begin{align}
    \MSTratio{\AF}{B}  &>  \frac{2 + \sqrt{3}}{\sqrt{3} + 3\ee} 
    >  2.154\ldots - 4\ee.
  \end{align}
  Since we can make $\ee > 0$ arbitrarily small, this implies the claimed lower bound.
\end{proof}

The example used to establish the lower bound can be extended to larger numbers of points, e.g.\ the following disjoint union of three lattices:
$B$ is the hexagonal lattice (to be defined shortly), and $\AF \setminus B$ is a slightly shifted copy of the hexagonal lattice, together with the barycenters of the triangles in every fourth row, which is a rectangular lattice with distances $1$ and $\sqrt{3}$ between consecutive rows and columns.

\smallskip
The question about the infimum of the maximum MST-ratio turns out to be less interesting, with $1$ as answer.
To see the lower bound, set $B=\AF$, in which case $\Length{\MST{B}} = \Length{\MST{\AF}}$ and $\Length{\MST{\AF \setminus B}} = 0$.
The ratio is therefore $1$.
We get the upper bound by constructing a set $\AF$ of $n \geq 2$ points.
It contains the origin, $n-2$ points each at distance at most $\ee > 0$ from the origin, and another point, which we call $b$, at unit distance from the origin.
Assume $b \in B$, and consider the case in which $B$ contains at least one other point of $\AF$.
Then
\begin{align}
    1  &\leq       ~~\,\Length{\MST{\AF}} ~~\, \leq  1 + 2(n-2) \ee , \\
    1 - \ee  &\leq ~~\,\Length{\MST{B}}   ~~\, \leq  1 + 2(n-2) \ee , \\
    0  &\leq    \Length{\MST{\AF \setminus B}} \leq  2(n-3) \ee .
\end{align}
For any given $\dd > 0$, we can choose $\ee > 0$ sufficiently small such that the ratio is smaller than $1 + \dd$.
In the other case, in which $B = \{b\}$, we have $\Length{\MST{B}} = 0$ and $\Length{\MST{\AF \setminus B}} \leq 2(n-2) \ee$, so we can make the ratio arbitrarily small and certainly smaller than $1$.

\section{Two-dimensional Lattices}
\label{sec:3}

Motivated by the triviality of the infimum maximum MST-ratio for general finite sets, we aim for a restriction that disallows extremely unbalanced distributions.
There are many choices, and we opt for a maximally restricted setting in which the MST-ratio is still an interesting question.
Specifically, we focus on $2$-dimensional lattices.
\begin{named}{Definition}
  The \emph{(Euclidean) lattice} spanned by two linearly independent vectors, $\uuu, \vvv \in \Rspace^2$, consists of all integer combinations of these vectors: $\AL (\uuu, \vvv)  =  \{ i \uuu + j \vvv \mid i,j \in \Zspace \}$.
\end{named}
By definition, lattices are infinite.
To cope with the difficulty of constructing the minimum spanning tree of infinitely many points, we take progressively larger but finite portions of a lattice and monitor the sequence of MST-ratios.
Specifically, we fix a partition of the infinite lattice, take rhombi centered at the origin and spanned by the vectors of the shortest basis of the lattice, for each rhombus get the MST-ratio for the points inside the rhombus, and consider the sequence of MST-ratios as the size of the rhombus increases.
If this sequence converges, we call the limit the \emph{MST-ratio} of the chosen partition of the lattice.
\begin{figure}[hbt]
    \centering \vspace{0.0in}
    \resizebox{!}{1.2in}{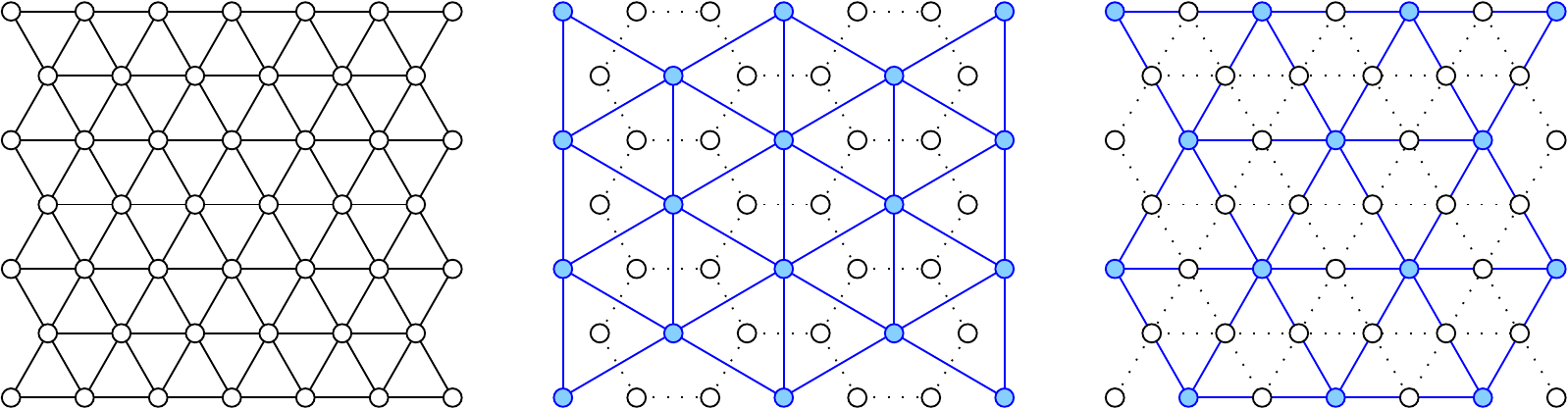}
    \vspace{-0.0in}
    \caption{\emph{Left:} a portion of the hexagonal lattice and all its shortest edges.
    \emph{Middle:} a partition into one and two thirds of the points, with MST-ratio converging to $(2+\sqrt{3})/3 = 1.245\ldots$.
    \emph{Right:} a partition into one and three quarters of the points, with MST-ratio converging to $1.25$.}
    \label{fig:hexagonal}
\end{figure}

A particularly interesting lattice is the \emph{triangular} or \emph{hexagonal lattice},
which is spanned by $\uuu = (1,0)$ and $\vvv = \frac{1}{2} (1, \sqrt{3})$; see the left panel in Figure~\ref{fig:hexagonal}.
The minimum distance between its points is $1$, so all edges of the MST have length $1$.
The two partitions illustrated in the middle and right panels of Figure~\ref{fig:hexagonal} have MST-ratios $1.245\ldots$ and $1.25$, respectively.
In one way or another, we use this lattice to prove all four bounds claimed in the following theorem.
\begin{theorem}
  \label{thm:maximum_MST-ratio_for_lattices}
  The supremum and infimum, over all $2$-dimensional lattices, $\AL$, of the supremum, over all subsets, $B \subseteq \AL$, of the MST-ratio are $\SupSup = \sup_{\AL} \sup_B \MSTratio{\AL}{B} = 2$ and $\InfSup = \inf_{\AL} \sup_B \MSTratio{\AL}{B} = 1.25$.
\end{theorem}
Each of the subsequent subsections restates and proves one of the four bounds, except for the last subsection, which only sketches the proof strategy, with the proof presented in Section~\ref{sec:4}.

\subsection{Lower Bound for Sup-Sup}
\label{sec:3.1}

This subsection exhibits a lattice, and a partition of this lattice into two sets, such that the MST-ratio of progressively larger finite portions of the lattice approaches the supremum of the supremum MST-ratio claimed in Theorem~\ref{thm:maximum_MST-ratio_for_lattices} from below.
\begin{claim}
  \label{clm:lower_bound_for_supmax}
  $\SupSup \geq 2$.
\end{claim}
\begin{proof}
  Let $\AL$ be the hexagonal lattice horizontally stretched by a factor $9$, and let $B \subseteq \AL$ be the one third of the points drawn blue in Figure~\ref{fig:stretched}.
  The (vertical) distance between neighboring points in a column of $\AL$ is $\sqrt{3}$, and the (horizontal) distance between two neighboring columns is $\frac{9}{2}$.
  \begin{figure}[hbt]
    \centering \vspace{0.0in}
    \resizebox{!}{2.0in}{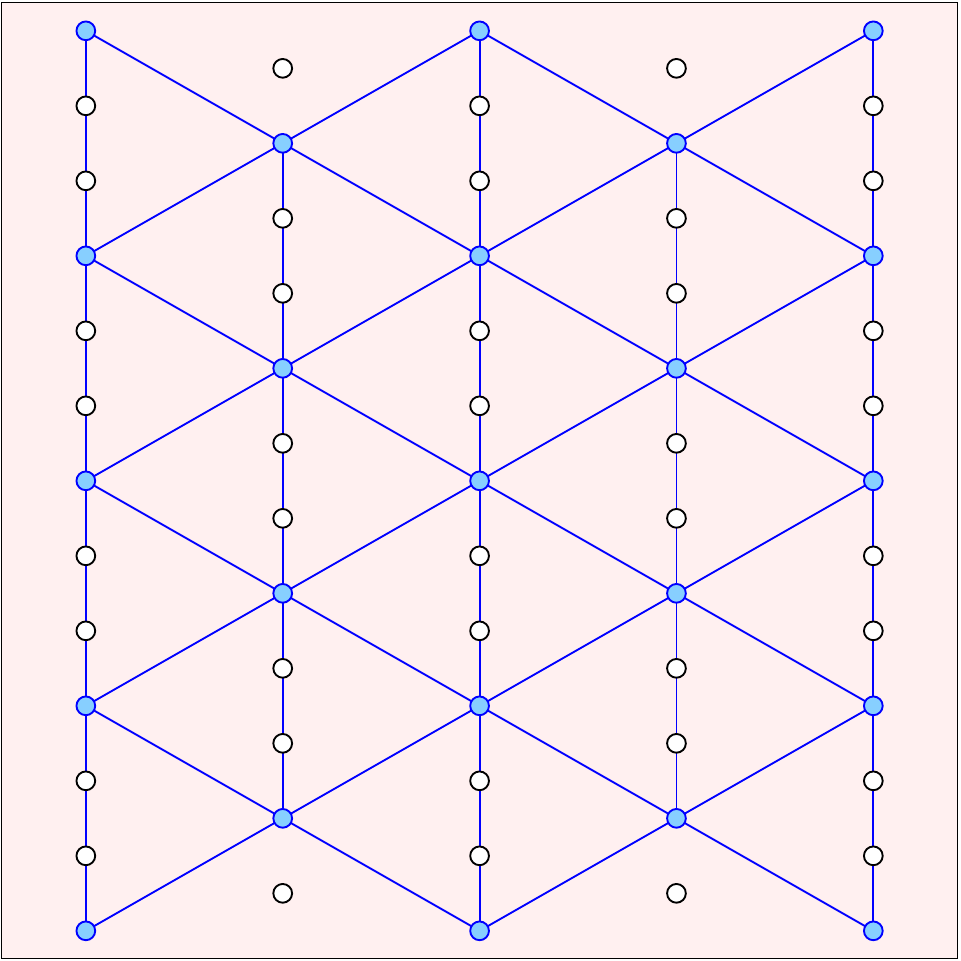}
    \vspace{-0.0in}
    \caption{The portion of the horizontally stretched hexagonal lattice, $\AL$, and the subset of \emph{blue} points, $B$, inside a square centered at the origin.
    The edges show the union of all possible minimum spanning trees of the \emph{blue} points.}
    \label{fig:stretched}
  \end{figure}
  For each $r \geq 0$, let $\AL_r \subseteq \AL$ and $B_r \subseteq B$ be the points in $[-r, r]^2$.
  Hence, $\AL_r$ consists of $p_r = 2 \lfloor {2r / 9} \rfloor + 1$ vertical columns, which alternate between $q_r = 2 \lfloor r/\sqrt{3} \rfloor + 1$ and $q_r - 1$ or $q_r + 1$ points.
  Observe that $p_r$ and $q_r$ are both odd, and that $n_r = q_r p_r \pm (p_r-1)/2$ is the cardinality of $\AL_r$.
  The number of points of $B_r$ in the columns alternates between $b_r = 2 \lfloor r/(3 \sqrt{3}) \rfloor + 1$ and $b_r - 1$ or $b_r + 1$, so $m_r = b_r p_r  \pm (p_r-1)$ is the cardinality of $B_r$.
  It is easy to see that $n_r - 2p_r \leq 3m_r \leq n_r + 2p_r$.

  \smallskip
  By choice of the stretch factor, $B$ is a hexagonal lattice with distance $3 \sqrt{3}$ between closest points.
  Hence, $\Length{\MST{B_r}} = 3 \sqrt{3} (m_r - 1)$.
  Compare this with a minimum spanning tree of $\AL_r$, which first connects the points in each column and second connects neighboring columns with one edge for each pair.
  Hence,
  \begin{align}
    \Length{\MST{\AL_r}} &= \sqrt{3} (n_r-p_r) + \sqrt{21} (p_r-1) ,
  \end{align}
  because every point, except the last in each column, has a neighbor at distance $\sqrt{3}$ below, and any two neighboring columns have points at distance $\sqrt{21}$ from each other.
  Similarly, any minimum spanning tree of $\AL_r \setminus B_r$ first connects the points in each column and second connects neighboring columns with one edge for each pair.
  Its length is therefore the same as that of $\MST{\AL_r}$.
  Using $3m_r = n_r + o(n_r)$, this implies
  \begin{align}
    \frac{\Length{\MST{B_r}} + \Length{\MST{\AL_r \setminus B_r}}}
         {\Length{\MST{\AL_r}}}
    &= \frac{3\sqrt{3}(m_r-1) + \sqrt{3}(n_r-p_r) + \sqrt{21}(p_r-1)}
            {\sqrt{3}(n_r-p_r) + \sqrt{21}(p_r-1)} \\
    &= \frac{2\sqrt{3}n_r + o(n_r)}
            {\sqrt{3}n_r + o(n_r)}
    \stackrel{r \to \infty}{\longrightarrow} 2 .
  \end{align}
  For any $\ee > 0$, we can choose $r$ sufficiently large such that the MST-ratio exceeds $2 - \ee$, which implies the claimed lower bound.
\end{proof}

\subsection{Upper Bound for Sup-Sup}
\label{sec:3.2}

This subsection proves the upper bound that matched the lower bound established in the preceding subsection.
Given any lattice and any partition of this lattice into two sets, we show that for any $\ee > 0$, the MST-ratio cannot exceed $2 + \ee$.
We begin with a bound for the length of the minimum spanning tree of any finite set in a square.
\begin{lemma}
  \label{lem:mst_insquare_upperbound}
  The length of the minimum spanning tree of any $n$ or fewer points in $[0,n]^2$ is at most $2 n \sqrt{n}$. 
\end{lemma}
\begin{proof}
  Assuming the number of points is $k \leq n$, the minimum spanning tree has $k-1$ edges, and we write $\ell_1, \ell_2, \ldots, \ell_{k-1}$ for their lengths.
  The sum of the squares of these lengths is at most $4n^2$, as proved in \cite{GiPo68}.
  By the Cauchy--Schwarz inequality, the sum of the $\ell_i$ is maximized when all terms are the same, namely $\ell_i^2 = 4n^2/(k-1)$ for all $i$.
  This implies 
  \begin{align}
    \sum\nolimits_{i=1}^{k-1} \ell_i
      &\leq (k-1) \sqrt{4n^2 / (k-1)}
       = 2n \sqrt{k-1} ,
  \end{align}
  from which the claimed bound follows.
\end{proof}
\Skip{As shown in \cite{AAABBCILSTY13}, the bound on the sum of squared edge lengths can be improved to $3.411\ldots n^2$.}
Lemma~\ref{lem:mst_insquare_upperbound} will provide a crucial step in the proof of the upper bound for the supremum maximum MST-ratio, which we present next.
\begin{claim}
\label{clm:upper_bound_for_supmax}
  $\SupSup \leq 2$.
\end{claim}
\begin{proof}
  We show that the MST-ratio of any lattice $\AL \subseteq \Rspace^2$ and any subset $B \subseteq \AL$ is at most the claimed upper bound.
  Let $\uuu$ be the shortest non-zero vector in $\Lambda$, and $\vvv$ the shortest non-zero vector that is not a multiple of $\uuu$,
  breaking ties arbitrarily if necessary.
  Suppose their lengths satisfy $1 = \norm{\uuu} \leq \norm{\vvv} = \nu$.
  To simplify language, we call the points on a line parallel to $\uuu$ a \emph{row} of $\AL$.
  For every positive integer, $n$, let $\AL_n \subseteq \AL$ contain all points $\alpha \uuu + \beta \vvv$, with $0 \leq \alpha, \beta \leq n$.
  The minimum spanning tree of $\AL_n$ first connects the points in each row and then the neighboring rows, so
  \begin{align}
    \Length{\MST{\AL_n}} &= (n+1)n + n \nu .
  \end{align}
  Set $B_n = B \cap \AL_n$.
  We construct a spanning tree, $T(B_n)$, by first connecting the points within the rows.
  This allows for the possibility that some rows do not contain any points of $B_n$.
  In each of the other rows, we choose an arbitrary but fixed point of $B_n$, write $B_n' \subseteq B_n$ for the chosen points, construct $\MST{B_n'}$, and add its edges to $T(B_n)$.
  Since $T(B_n)$ spans $B_n$ but is not necessarily the shortest such tree, so $\Length{\MST{B_n}} \leq \Length{T(B_n)}$.
  To bound the latter, recall that there are $n+1$ rows, each of length at most $n$.
  Furthermore, $B_n'$ consists of at most $n+1$ points that fit inside a square of side length $n (\nu+1)$, in which $\nu$ is independent of $n$.
  Lemma~\ref{lem:mst_insquare_upperbound} implies $\Length{\MST{B_n'}} \leq 2 (\nu+1) \sqrt{\nu+1} \cdot n \sqrt{n}$.
  Hence,
  \begin{align}
    \Length{\MST{B_n}} &\leq (n+1)n + 2(\nu+1)\sqrt{\nu+1} \cdot n\sqrt{n} .
  \end{align}
  By symmetry, we have the same upper bound for the length of $\MST{\AL_n \setminus B_n}$.
  Comparing this with the minimum spanning tree of $\AL_n$, we get
  \begin{align}
    \frac{\Length{\MST{B_n}} + \Length{\MST{\AL_n\setminus B_n}}}{\Length{\MST{\AL_n}}} 
    \leq \frac{2n^2+2n + 4 (\nu+1)^{3/2} \cdot n\sqrt{n}}
              {n^2 + n + \nu n}
    \stackrel{n \to \infty}{\longrightarrow} 2.
    \label{eqn:mst-ratio-upperbound}
  \end{align}
  For every $\ee > 0$, we can choose $n$ large enough so that the MST-ratio is less than $2 + \ee$.
  This works for every lattice and partition, which implies the claimed upper bound.
\end{proof}

\subsection{Lower Bound for Inf-Sup}
\label{sec:3.3}

This subsection establishes the lower bound for the infimum, over all lattices, of the supremum MST-ratio.
We do this by establishing a partition into one and three quarters that can be defined for any lattice and has MST-ratio at least as large as claimed in Theorem~\ref{thm:maximum_MST-ratio_for_lattices}.
\begin{claim}
  \label{clm:lower_bound_for_infmax}
  $\InfSup \geq 1.25$.
\end{claim}
\begin{proof}
  Let $\uuu$ and $\vvv$ be two vectors spanning $\AL$, and let $B$ be the sublattice spanned by $2 \uuu$ and $2 \vvv$.
  Assuming the minimum distance between two points in $\AL$ is $1$, most edges of $\MST{\AL}$ have length $1$, while most edges of $\MST{B}$ have length $2$.
  Write $\AL_n \subseteq \AL$ for the points $i\uuu+j\vvv$, with $-2n \leq i,j \leq 2n+1$, and $B_n \subseteq \AL_n$ for the points with even $i$ and $j$.
  Since $B_n$ contains only a quarter of the points, this implies $\lim_{n \to \infty} \Length{\MST{B_n}} / \Length{\MST{\AL_n}} = \frac{1}{2}$.
  The complement of $B_n$ contains three quarters of the points, and the edges in its minimum spanning tree have length at least $1$, which implies $\lim_{n \to \infty} \Length{\MST{\AL_n \setminus B_n}} / \Length{\MST{\AL_n}} \geq \frac{3}{4}$.
  Hence, the MST-ratio of $B \subseteq \AL$ is at least $\frac{1}{2} + \frac{3}{4} = 1.25$.
\end{proof}

\subsection{Upper Bound for Inf-Sup}
\label{sec:3.4}

The upper bound for the infimum of the supremum MST-ratio will be proved in Section~\ref{sec:4}.
This proof is carefully constructed from a network of inequalities that require attention to detail.
This subsection makes an argument why it is not unreasonable to believe that significant short-cuts may be difficult to find.

\begin{figure}[hbt]
    \centering \vspace{0.05in}
    \resizebox{!}{3.4in}{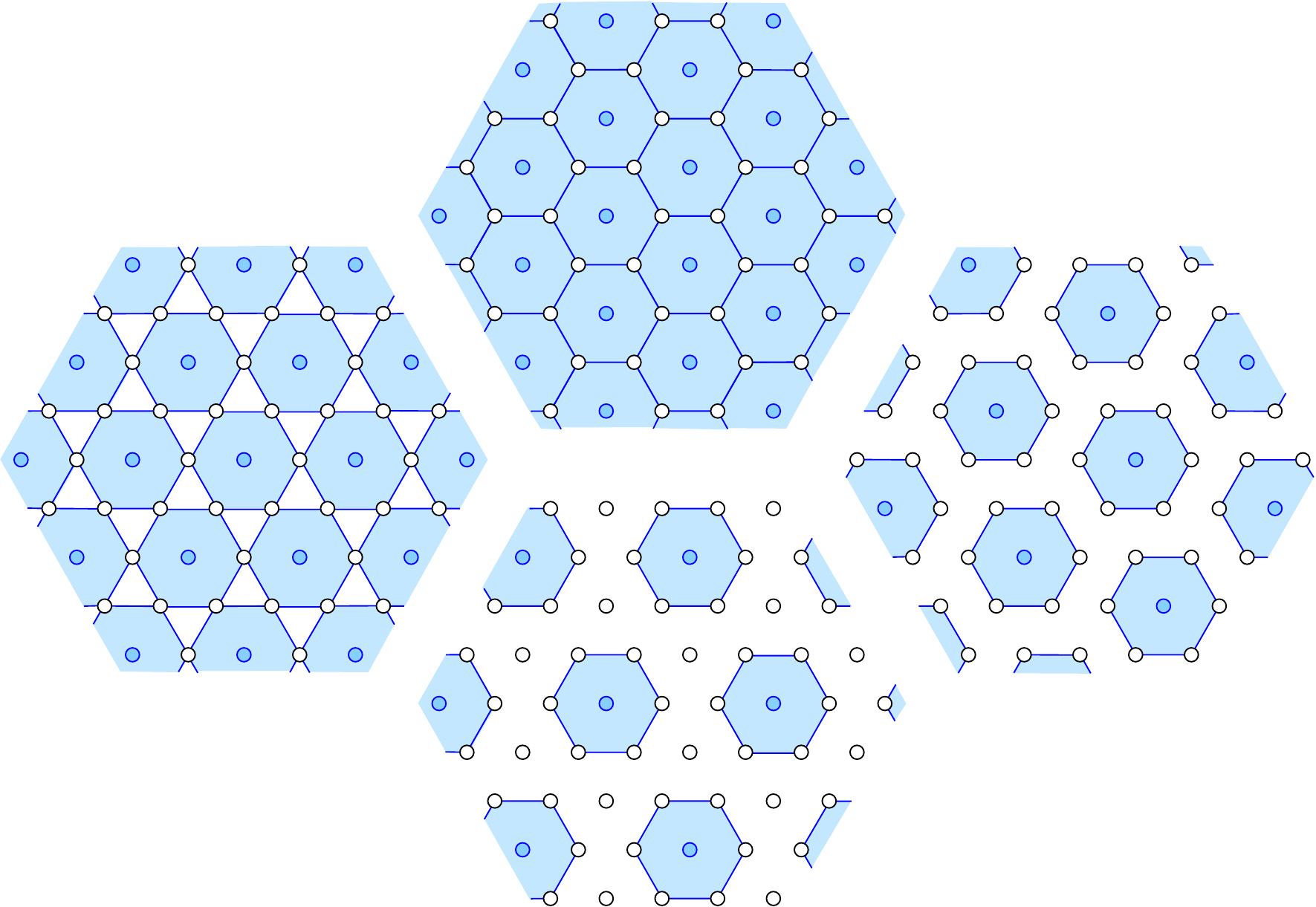}
    \vspace{-0.0in}
    \caption{Four partitions of the hexagonal lattice into two sets, in which we draw each \emph{(blue)} point of the smaller set with its hexagonal neighborhood.
    The proportions of \emph{blue} versus \emph{white} points are
    $1:2$ in the \emph{upper middle}, $1:3$ on the \emph{left}, $1:6$ on the \emph{right}, and $1:8$ in the \emph{lower middle}.
    The corresponding MST-ratios are approximately $1.245$, $1.25$, $1.236$, and $1.222$, in this sequence.}
    \label{fig:fourpartitions}
\end{figure}
The lattice that is most resistant to large MST-ratios is the hexagonal lattice, $\AL$, of which four different subsets, $B \subseteq \AL$, are illustrated as packings of hexagonal neighborhoods in Figure~\ref{fig:fourpartitions}.
Starting at the upper middle, then left, then right, and finally the lower middle, the density of the packing decreases monotonically as the minimum distance between points of $B$ increases from $\sqrt{3}$ to $2$, to $\sqrt{7}$, and finally to $3$.
Correspondingly, $B$ contains one third, one quarter, one seventh, and one ninth of the points.
Perhaps surprisingly, the MST-ratio does not vary monotonically and attains the largest value for the subset $B$ that contains one quarter of the points.
The purpose of Section~\ref{sec:4} is to prove that no other subset of $\AL$ achieves a larger MST-ratio; that is: $1.25$ is the supremum MST-ratio of the hexagonal lattice.
\begin{claim}
  \label{clm:upper_bound_for_infmax}
  $\InfSup \leq 1.25$.
\end{claim}
Because the value matches the lower bound stated in Claim~\ref{clm:lower_bound_for_infmax}, this implies that $1.25$ is indeed the infimum, over all $2$-dimensional lattices, of the supremum MST-ratio.
Prior to studying the hexagonal lattice, the authors of this paper proved that the supremum MST-ratio of the integer lattice is $\sqrt{2}$---which happens to match the ratio found for random sets \cite{DPT23}---and the optimizing subset are the points whose coordinates add up to even integers.
The proof is similar to the one for the hexagonal lattice presented in Section~\ref{sec:4}, and almost as long.
If instead we consider the points whose coordinates add up to odd integers, we get the same MST-ratio, so the integer lattice has at least two globally optimal partitions that are far from each other if the difference is measured in terms of the color changes needed to turn one into the other.
Similarly, the hexagonal lattice has at least four globally optimal partitions, and moving from one to the other (by flipping colors) means walking a path along which the MST-ratio is sometimes barely below $1.25$.
To support the hypothesis of a rugged but shallow landscape, we conducted computational experiments for finite subsets of the integer lattice, which identified many local maxima that prevent local improvement strategies from reaching any global maximum.
We feel that these findings justify the exhaustive case analysis in Section~\ref{sec:4}, and the many delicate inequalities in that section give evidence for how close the paths get to the supremum MST-ratio.

\section{Hexagonal Lattice on Torus}
\label{sec:4}

In this section, we prove Claim~\ref{clm:upper_bound_for_infmax} for the hexagonal lattice on the torus.
We begin by constructing this lattice from a portion of the hexagonal lattice in the plane and proving that the minimum spanning trees in the two topologies are not very different in length.
In the remaining subsections, we give a precise statement of the theorem that implies Claim~\ref{clm:upper_bound_for_infmax} and prove the theorem with a packing argument in six steps.

\subsection{Plane versus Torus}
\label{sec:4.1}

We consider the hexagonal lattice on the torus rather than in $\Rspace^2$ in order to eliminate boundary effects, which appear when we study a finite portion of the hexagonal lattice.
Let $\uuu$ and $\vvv$ be two unit vectors with a $60^\circ$ degree angle between them, and write $\AL \subseteq \Rspace^2$ for the hexagonal lattice they span.
For every positive $n \in \Zspace$, let $\AL_n \subseteq \AL$ contain the $n^2$ points $a = \alpha \uuu + \beta \vvv$ with $0 \leq \alpha, \beta \leq n-1$.
We write $\AL_n'$ for the same $n^2$ points but with the topology of the torus, which we get by identifying $a$ with $a + in\uuu + jn\vvv$ for all $i, j \in \Zspace$, and defining the distance as the minimum Euclidean distance between any two representatives.
\begin{figure}[hbt]
    \centering \vspace{0.1in}
    \resizebox{!}{1.4in}{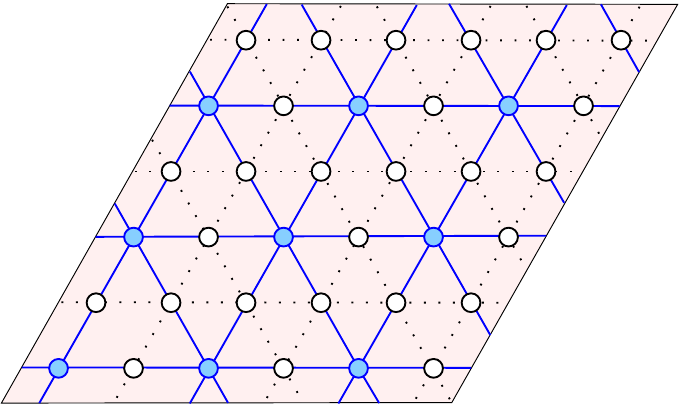}
    \vspace{-0.0in}
    \caption{The hexagonal lattice of $36$ points on the torus, obtained by gluing opposite sides of the rhombus.
    The sublattice with twice the distance between neighboring points in shown in \emph{blue}.}
    \label{fig:rhombus}
\end{figure}
Equivalently, consider the rhombus of points $\varphi \uuu + \psi \vvv$ for real coefficients $- \frac{1}{2} \leq \varphi, \psi \leq n-\frac{1}{2}$, and glue this rhombus along opposite sides as illustrated for $n = 6$ in Figure~\ref{fig:rhombus}.
Call the boundary of this rhombus the \emph{seam}.
Its length is $4n$ in the plane but only $2n$ on the torus since the sides are glued in pairs.
Note also that every point of $\AL$ has distance at least $\sqrt{3} / 4$ from the nearest point in the seam.
\begin{lemma}
  \label{lem:plane_versus_torus}
  Let $\AL \subseteq \Rspace^2$ be the hexagonal lattice, $\AL_n \subseteq \AL$ the subset of $n^2$ points, and $\AL_n'$ the same $n^2$ points but on the torus, as described above.
  For any subset $B_n \subseteq \AL_n$ and the corresponding subset $B_n' \subseteq \AL_n'$ on the torus, the lengths of the minimum spanning trees satisfy $\Length{\MST{B_n'}} \leq \Length{\MST{B_n}} \leq \Length{\MST{B_n'}} + 32 \sqrt{2} \cdot n \sqrt{n}$.
\end{lemma}
\begin{proof}
  Fix two minimum spanning trees, $T$ of $B_n$ in $\Rspace^2$ and $T'$ of $B_n'$ on the torus.
  Since the distances on the torus are smaller than or equal to those in $\Rspace^2$, we have $\Length{T'} \leq \Length{T}$, which is the first claimed inequality.
  Let $E'$ be the edges of $T'$ that have the same length in both topologies, and let $E''$ be the other edges of $T'$, which are shorter on the torus than in $\Rspace^2$.
  To draw an edge of $E''$ in the plane so its length matches the length on the torus, we need to connect representatives of the endpoints that lie in different rhombi.
  Assuming one endpoint is in $\AL_n$, this edge crosses the seam.
  In contrast, every edge in $E'$ can be drawn between two points of $\AL_n$, so without crossing the seam.
  We will prove shortly that the distance between two crossings measured along the seam is at least $\frac{1}{2}$.
  Since the length of the seam is $2n$, this implies that $E''$ contains at most $4n$ edges.
  Let $V'' \subseteq \AL_n$ be the set of at most $8n$ endpoints of the edges in $E''$, and let $T''$ be a minimum spanning tree of $V''$, with distances measured in $\Rspace^2$.
  Since $\AL_n$ easily fits inside a square with sides of length $8n$, Lemma~\ref{lem:mst_insquare_upperbound} implies $\Length{T''} \leq 32\sqrt{2} \cdot n \sqrt{n}$.
  The edges in $E'$ together with the edges of $T''$ form a connected graph with vertices $\AL_n$.
  Hence,
  \begin{align}
    \Length{T} &\leq \Length{T'} + \Length{T''}
                \leq \Length{T'} + 32 \sqrt{2} \cdot n \sqrt{n} ,
  \end{align}
  which is the second claimed inequality.
  It remains to show that the distance between two crossings along the seam is at least $\frac{1}{2}$.
  Let $ab$ and $xy$ be two edges in $E''$, and recall that the greedy construction of the minimum spanning tree prohibits $x$ and $y$ to lie inside the smallest circle that passes through $a$ and $b$, and vice versa.
  If the edges share an endpoint, then the angle between them is at least $60^\circ$.
  Since the common endpoint is at distance at least $\sqrt{3}/4$ from the seam, this implies the claimed lower bound on the distance between the two crossings.
  So assume $a, b, x, y$ are distinct, and let $c \in ab$ and $z \in xy$ be the points that minimize the distance between the edges, and observe that $\Edist{c}{z}$ is a lower bound for the distance between the crossings.
  At least one of $c$ and $z$ must be an endpoint, so suppose $z = x$.
  But since $x$ lies outside the smallest circle of $a$ and $b$, and outside the unit circles centered at $a$ and $b$, the distance of $x$ to any point of $ab$ is at least $1$.
\end{proof}

The inequalities in Lemma~\ref{lem:mst_insquare_upperbound} generalize to all $2$-dimensional lattices.
Letting $\uuu$ and $\vvv$ be two shortest vectors that span a lattice, and assuming $1 = \norm{\uuu} \leq \norm{\vvv} = \nu$, we get $2 (4+4\nu)^{3/2} \cdot n \sqrt{n}$ as an upper bound for the difference in length, in which we compare a rhombus of $n \times n$ points in $\Rspace^2$ and on the torus, as before.

\subsection{Statement of Theorem}
\label{sec:4.2}

We fix $n$ to an even integer and write $\BL = \AL_n'$ for the hexagonal lattice on the torus.
Since $n$ is even, $\BL_1 = \{2x \mid x \in \BL\}$ is a hexagonal sublattice of $\BL$, and we set $\BL_3 = \BL \setminus \BL_1$; see Figure~\ref{fig:rhombus}.
The lengths of the three minimum spanning trees are easy to determine because they use only the shortest available edges, which have length $1$ for $\BL$ and $\BL_3$, and length $2$ for $\BL_1$.
The MST-ratio is therefore
\begin{align}
  \MSTratio{\BL}{\BL_1}  &=  \frac{\Length{\MST{\BL_1}} + \Length{\MST{\BL_3}}}{\Length{\MST{\BL}}} 
  =  \frac{2 \left( n^2/4-1 \right) + \left( 3n^2/4-1 \right)} {n^2-1}
  \stackrel{n \to \infty}{\longrightarrow} 1.25.
  \label{eqn:minMSTratiohex}
\end{align}
Call an edge \emph{short} if its length is $1$.
All other edges have length larger than the desired average, which is $\frac{5}{4} = 1.25$, so we call them \emph{long}.
While $\MST{\BL_3}$ has only short edges, and $\MST{\BL_1}$ uses only the shortest edges connecting its points, we claim that their combined length is as large as it can be.
\begin{theorem}
  \label{thm:maximum_MST-ratio_of_hexagonal_lattice}
  Let $\BL$ be a hexagonal lattice on the torus.
  Then the maximum MST-ratio of $\BL$ converges to $\frac{5}{4} = 1.25$ from below.
\end{theorem}
The proof consists of six steps, which are presented in the same number of subsections:
\ref{sec:4.3} introduces the hexagonal distance, compares its MST with the Euclidean MST, and uses the former to formulate the proof strategy;
\ref{sec:4.4} introduces the main tool, which are hexagonal-neighborhoods of the lattice points;
\ref{sec:4.5} constructs a hierarchy of such neighborhoods aimed at counting the short edges;
\ref{sec:4.6} introduces so-called satellites, which provide additional short edges needed in the proof;
\ref{sec:4.7} forms loop-free subgraphs of short edges and bounds their sizes;
and \ref{sec:4.8} does the final accounting while paying special attention to the cases in which all long edges have length between $\sqrt{3}$ and $3$.
Throughout this proof, we use the fact that the minimum spanning tree can be computed by greedily adding the shortest available edge that does not form a cycle to the tree \cite{Bor26,Kru56}.

\subsection{Hexagonal Distance and Proof Strategy}
\label{sec:4.3}

It is convenient to write the points in $\BL$ with three integer coordinates.
To explain this, let
\begin{align}
  \xxx &= \tfrac{1}{\sqrt{3}} \left( 0, 1 \right) , ~~
  \yyy  = \tfrac{1}{\sqrt{3}} \left( - \tfrac{\sqrt{3}}{2}, - \tfrac{1}{2} \right) , ~~
  \zzz  = \tfrac{1}{\sqrt{3}} \left( \tfrac{\sqrt{3}}{2}, - \tfrac{1}{2} \right) 
\end{align}
be three vectors, each of length $\sqrt{3}/3$, that mutually enclose an angle of $120^\circ$.
These are the projections of the unit coordinate vectors of $\Rspace^3$ onto the plane normal to the diagonal direction, scaled such that the three points are mutually one unit of distance apart.
The plane consists of all points $u = a \xxx + b \yyy + c \zzz$ for which $a+b+c = 0$, and such a point belongs to the hexagonal lattice iff $a,b,c \in \Zspace$; see Figure~\ref{fig:unitballhex}.
Given a second point, $v = \alpha \xxx + \beta \yyy + \gamma \zzz$, we write $i = a-\alpha$, $j = b-\beta$, $k = c-\gamma$ to compute the squared Euclidean distance between $u$ and $v$.
Since $\xxx^2 = \yyy^2 = \zzz^2 = \frac{1}{3}$ and $\xxx \yyy = \yyy \zzz = \zzz \xxx = - \frac{1}{6}$, we get
\begin{align}
  \Edist{u}{v}^2 &= 
  \norm{i\xxx + j\yyy + k\zzz}^2
     =  \tfrac{1}{3} (i^2 + j^2 + k^2) - \tfrac{1}{3} (ij + ik + jk)
     =  i^2 +ij +j^2 ,
\end{align}
in which we get the final expression using $k = -(i+j)$.
For points of the hexagonal lattice, $i$ and $j$ are integers, and so is the squared Euclidean distance between them.
It follows that the minimum distance between two points in $\BL$ is $1$.
\begin{figure}[hbt]
    \centering \vspace{0.10in}
    \resizebox{!}{1.2in}{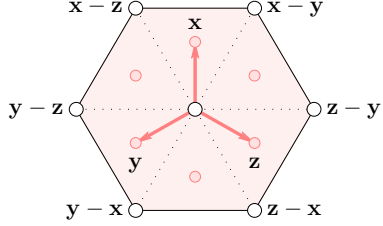}
    \vspace{-0.05in}
    \caption{The unit disk under the hexagonal distance in the plane.
    The edges that connect the origin to the corners at $\pm (\xxx-\yyy)$, $\pm (\yyy-\zzz)$, $\pm (\zzz-\xxx)$ decompose the hexagon into six equilateral triangles, whose barycenters are $\pm \xxx$, $\pm \yyy$, $\pm\zzz$.}
    \label{fig:unitballhex}
\end{figure}

\smallskip
We adapt the notion of distance to construct neighborhoods in the hexagonal lattice.
By definition, the \emph{hexagonal distance} between points $u = a \xxx + b \yyy + c \zzz$ and $v = \alpha \xxx + \beta \yyy + \gamma \zzz$ is
\begin{align}
  \hexdist{u}{v}  &=  \max \{ |a-\alpha|, |b-\beta|, |c-\gamma| \}  =  \max \{ |i|, |j|, |i+j| \} .
\end{align}
The \emph{unit disk} under this distance consists of all points with hexagonal distance at most $1$ from the origin:
$\Hdisk = \{ u \in \Rspace^2 \mid \hexdist{u}{0} \leq 1 \}$.
It is the regular hexagon with unit length sides that is the convex hull of the points $\pm (\xxx-\yyy)$, $\pm (\yyy-\zzz)$, $\pm (\zzz-\xxx)$; see Figure~\ref{fig:unitballhex}.
For $B \subseteq \BL$, we write $\hexMST{B}$ for the spanning tree that minimizes the hexagonal length.
We construct it by adding the edges in sequence of non-decreasing hexagonal length, breaking ties with Euclidean length, and breaking the remaining ties arbitrarily.
Since $\hexMST{B}$ is a spanning tree but not necessarily the one that minimizes Euclidean length, we have
\begin{align}
  \Length{\MST{B}}  &\leq  \Length{\hexMST{B}} ,
    \label{eqn:Euclideanhex}
\end{align}
in which we measure the Euclidean length on both sides.
To prove Theorem~\ref{thm:maximum_MST-ratio_of_hexagonal_lattice}, we show that for every $B \subseteq \BL$, the average (Euclidean) length of the long edges in $\hexMST{B}$ and the short edges in $\hexMST{\BL \setminus B}$ is at most $\frac{5}{4}$.
Interchanging $B$ and $\BL \setminus B$, we get the same relation by symmetry.
Using \eqref{eqn:Euclideanhex}, this implies
\begin{align}
  \Length{\MST{B}} + \Length{\MST{\BL \setminus B}}
    &\leq \Length{\hexMST{B}} + \Length{\hexMST{\BL \setminus B}} \leq \tfrac{5}{4} (n^2 - 2) .
  \label{eqn:UpperBound}
\end{align}
Compare this with \eqref{eqn:minMSTratiohex}, which establishes $\Length{\MST{\BL_1}} + \Length{\MST{\BL_3}} = \frac{5}{4}n^2 - 3$ for the partition $\BL = \BL_1 \sqcup \BL_3$.
The right-hand side differs from the upper bound in \eqref{eqn:UpperBound} by only a small additive constant.
We thus conclude that the maximum MST-ratio of $\BL$ converges to $\frac{5}{4}$ from below, as claimed by Theorem~\ref{thm:maximum_MST-ratio_of_hexagonal_lattice}.

\subsection{Hierarchy of Habitats}
\label{sec:4.4}

Let $\Tree{\ell}{B}$ be the subset of edges in $\hexMST{B}$ whose hexagonal lengths are at most $\ell$, together with the endpoints of these edges.
For example, $\Tree{0}{B}$ has zero edges, $\Tree{1}{B}$ consist of all short edges, and $\Tree{\ell}{B} = \hexMST{B}$ for sufficiently large $\ell$.
All edges connecting points in different components of $\Tree{\ell}{B}$ have hexagonal length $\ell+1$ or larger.
We thus write $k\Hdisk$ for the scaled copy of the unit disk and call
\begin{align}
  \Thicken{k}{B}  &=  \bigcup\nolimits_{u \in B} (k\Hdisk +u)
\end{align}
the \emph{$k$-th thickening} of $B$, in which $k\Hdisk + u$ is the translate of $k\Hdisk$ whose center is $u$.
As illustrated in Figure~\ref{fig:thickenhex}, the $k$-th thickenings of points $u$ and $v$ \emph{overlap}, \emph{touch}, are \emph{disjoint} if the hexagonal distance between $u$ and $v$ is less than, equal to, larger than $2k$, respectively.
\begin{figure}[hbt]
    \centering \vspace{0.1in}
    \resizebox{!}{2.4in}{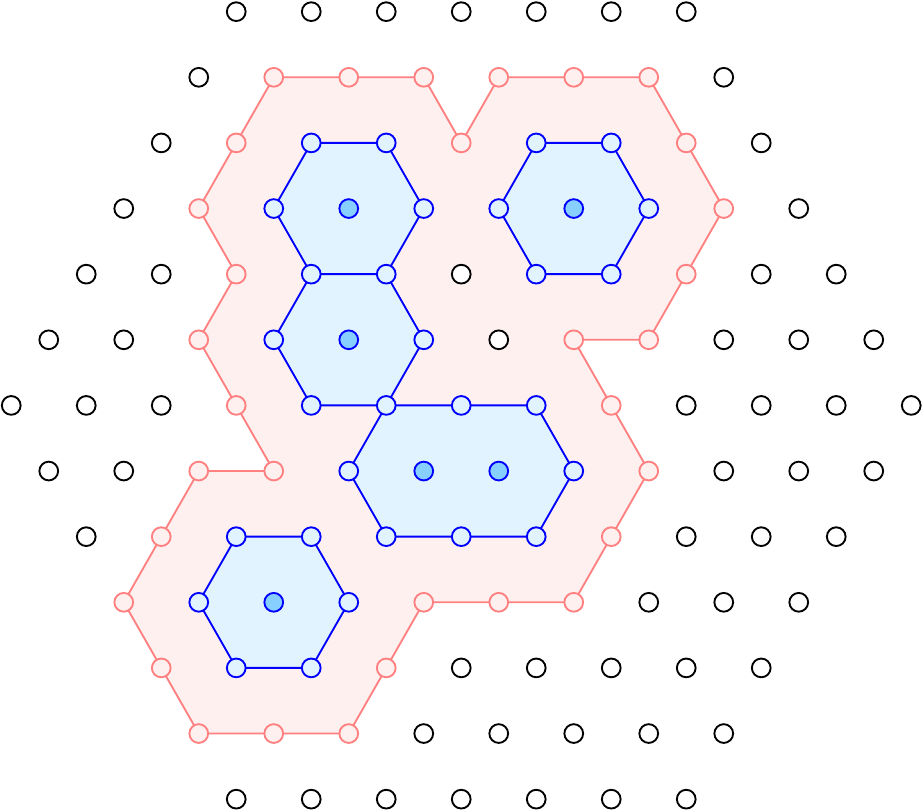}
    \vspace{-0.0in}
    \caption{The \emph{blue} $1$-st thickening and the \emph{pink} $2$-nd thickening of $B = \{{\tt a}, {\tt b}, {\tt c}, {\tt d}, {\tt e}, {\tt f}\}$ in the hexagonal lattice.
    $\Hdisk+{\tt a}$ and $\Hdisk+{\tt b}$ share an edge and therefore form two rooms in a common house, while $\Hdisk+{\tt e}$ and $\Hdisk+{\tt f}$ overlap and thus form a one-room house in $\Thicken{1}{B}$.
    These two houses form a block, and together with $\Hdisk+{\tt d}$, they form a compound of two blocks.
    $\Hdisk+{\tt c}$ is a room, a house, a block, and a compound by itself.
    The two compounds lie in the interior of a room in $\Thicken{2}{B}$.}
    \label{fig:thickenhex}
\end{figure}

The boundary of $k\Hdisk$ passes through $6k$ points of the hexagonal lattice, which we call the \emph{vertices} of $k\Hdisk$.
Furthermore, we call the $6k$ (short) edges that connect these points in cyclic order the \emph{edges} of $k\Hdisk$.
Let $B_k \subseteq B$ be the vertex set of a component of $\Tree{2k-1}{B}$, and observe that for all $u,v \in B_k$ there is a sequence of points $u = x_1, x_2, \ldots, x_m = v$ in $B_k$ such that $k\Hdisk+x_i$ and $k\Hdisk+x_{i+1}$ overlap for all $1 \leq i \leq m-1$.
We define the \emph{frontier} of the component, denoted $\Frontier{k}{B_k}$, as the lattice points and the connecting (short) edges in the boundary of $\Thicken{k}{B_k}$.
Furthermore, $\Frontier{k}{B}$ is the union of frontiers of the components of $\Tree{2k-1}{B}$.
These notions are illustrated in Figure~\ref{fig:thickenhex}, which shows $\Frontier{1}{B}$ and $\Frontier{2}{B}$ for six marked points.
Note that the edge shared by $\Hdisk + {\tt a}$ and $\Hdisk + {\tt b}$ is part of $\Frontier{1}{B}$.

\subsection{Subdivided Foreground and Background}
\label{sec:4.5}

Consider the $1$-st thickening of $B$, which for the time being we call the \emph{foreground}.
Letting $B_1 \subseteq B_2$ be the vertex sets of two nested components of $\Tree{1}{B}$ and $\Tree{2}{B}$, we call $\Thicken{1}{B_1}$ a \emph{room} and $\Thicken{1}{B_2}$ a \emph{block} of the foreground.
We say two rooms are \emph{adjacent} if they share at least one edge.
In Figure~\ref{fig:thickenhex}, there are five rooms, two of which are adjacent, and three blocks, one of which contains three rooms.

\smallskip
To make a finer distinction, observe that for any edge, its Euclidean length is smaller than or equal to the hexagonal length.
The two notions agree on edges with slope $0$, $\sqrt{3}$, and $-\sqrt{3}$.
Consider $\Tree{2}{B}$ and $\Tree{3}{B}$ after removing all edges whose Euclidean length equals $2$ and $3$, respectively, and let $B_2'$ and $B_3'$ be the vertex sets of the components that satisfy $B_1 \subseteq B_2' \subseteq B_2 \subseteq B_3'$.
Observe that any two rooms in $\Thicken{1}{B_2'}$ have a sequence of pairwise adjacent rooms connecting them.
We therefore call $\Thicken{1}{B_2'}$ a \emph{house}.
For comparison, any two rooms in $\Thicken{1}{B_2}$ have a sequence of rooms connecting them such that any two consecutive rooms share at least a vertex but not necessarily a full edge.
Similarly, for any two blocks in $\Thicken{1}{B_3'}$, there is a sequence of blocks connecting them such that the channel separating any two consecutive blocks at its narrowest place is only $\sqrt{3}/2$ wide.
We therefore call $\Thicken{1}{B_3'}$ a \emph{compound}; see Figure~\ref{fig:thickenhex} for examples.
For comparison, the channel that separates two compounds is at its narrowest place at least one unit of distance wide.
A few observations:
\medskip \begin{description}
  \item[(i)] all vertices of $\Frontier{1}{B}$ are points in $\BL \setminus B$;
  \item[(ii)] all edges of $\Frontier{1}{B}$ are short;
  \item[(iii)] the frontier of a room consists of at least six (short) edges.
\end{description} \medskip
We call the complement of the foreground the \emph{background}, and the components of the background its \emph{backyards}.
We say a backyard is \emph{adjacent} to a house if the two share a non-empty portion of their boundary.
There are configurations in which the number of backyards is twice the number of houses; see Figure~\ref{fig:fourpartitions} on the left, where each backyard is adjacent to three houses, and each house is adjacent to six backyards.
In general, we distinguish between backyards adjacent to at most two and at least three houses, denoting their numbers $\alpha_1$ and $\beta_1$, respectively.
We prove an upper bound for $\beta_1$ in terms of the number of houses and blocks.
\begin{lemma}
  \label{lem:number_of_backyards}
    Given $h_1$ houses arranged in $b_1$ blocks, the number of backyards adjacent to three or more houses satisfies $\beta_1 \leq 2h_1 - 2b_1 + 2$.
\end{lemma}
\begin{proof}
  We construct a graph $G = G(B)$ on the torus by placing a node inside each house, and whenever two houses meet at a boundary vertex, we connect the corresponding nodes with a curved arc that passes through the shared vertex.
  This can be done such that no two of the arcs cross and each face of $G$ contains one backyard.
  A face bounded by a single arc (\emph{loop}) or two arcs (\emph{multi-arcs}) contains a backyard adjacent to at most two houses and thus does not count toward $\beta_1$.
  We remove this face by deleting the loop or one of the two multi-arcs.
  The resulting graph has $h_1$ nodes, $b_1$ components, and $\beta_1$ faces.
  Write $a_1$ for the number of arcs.
  If the graph is connected and all faces are bounded by three arcs, we have $h_1 - a_1 + \beta_1 = 0$ because the Euler characteristic of the torus is $0$.
  Whenever we remove an arc from this graph, we either merge two faces or split a component, but it is also possible that the removal of the arc has neither of those two side-effects.
  Hence, we have $h_1 - a_1 + \beta_1 \geq b_1 - 1$ in the general case.
  Since $2a_1 \geq 3\beta_1$, this implies $\beta_1 \leq 2h_1 - 2b_1 + 2$, as claimed.
\end{proof}

\subsection{Satellites}
\label{sec:4.6}

By definition, compounds cannot be packed as tightly as blocks; see Figure~\ref{fig:fourpartitions} with lattice points between the compounds in the lower middle but no such points between the blocks on the right.
Recall that each component of $\Thicken{1}{B}$ is contained in a room of $\Thicken{2}{B}$.
For each such room, we single out the largest compound it contains---breaking ties arbitrarily---and call this the \emph{big compound} of the room.
All others are \emph{small compounds} of the room.
We refer to certain lattice points close to one or more compounds as satellites.
The targeted lattice points are at distance $\sqrt{3}/2$ outside $\Thicken{1}{B}$ and either on the boundary or in the interior of $\Thicken{2}{B}$.
\begin{figure}[hbt]
    \centering \vspace{0.05in}
    \resizebox{!}{1.1in}{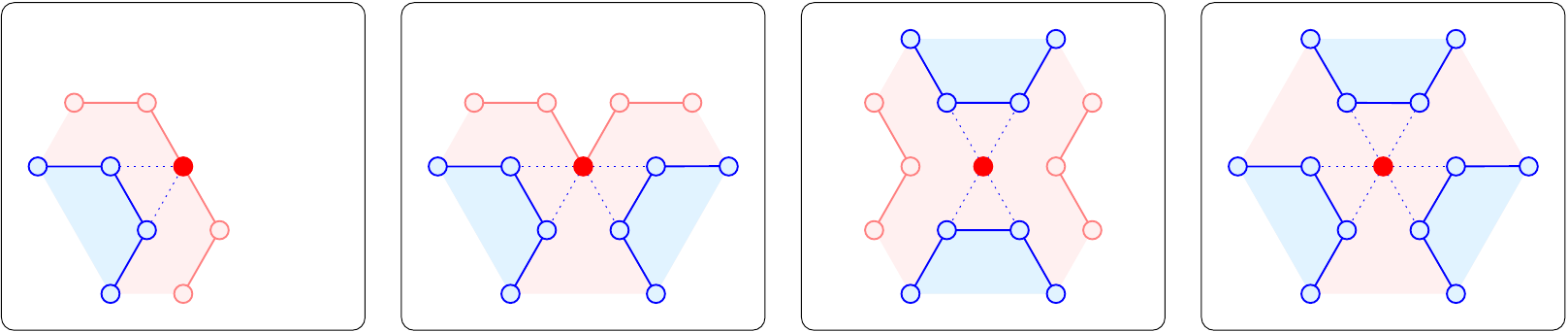}
    \vspace{-0.0in}
    \caption{From \emph{left} to \emph{right}: a single, a double, another double, and a triple satellite in \emph{red}. 
    In the \emph{left} two cases, the satellite belongs to the frontier of a room of the $2$-nd thickening of $B$, while in the \emph{right} two cases, the satellite lies in the interior of such a room.}
    \label{fig:satellite}
\end{figure}

The difference between small and large compounds influences which lattice points we call satellites.
For each small compound we find three satellites as follows:  sandwich the compound between three lines with slopes $0, \pm\sqrt{3}$, choose a (short) edge as the basis of an equilateral triangle outside the compound on each line, and pick the vertex of this triangle opposite to the basis as a \emph{satellite}.
Observe that the Euclidean distance between any two satellites of the same compound is at least $3$.
In contrast, we pick six lattice points as the satellites of the big compound by sandwiching it between six lines, two each of slope $0, \pm\sqrt{3}$, choosing one basis on each line, and picking the vertex of the equilateral triangle opposite to the basis as a satellite.
The Euclidean distance between any two such satellites is at least $\sqrt{3}$.

\smallskip
As illustrated in Figure~\ref{fig:satellite}, a lattice point can be a satellite of one, two, or three compounds in the same room.
Accordingly, we call the point a \emph{single}, \emph{double}, or \emph{triple satellite} of the room, respectively.
A single satellite is necessarily a vertex on the frontier of the room, a triple satellite is necessarily in the interior of the room, and a double satellite can be one or the other.
For a room, $R$, we write $s(R)$ and $d(R)$ for the number of single and double satellites on its frontier, and $e(R)$ and $t(R)$ for the number of double and triple satellites in its interior.
Summing over all rooms in $\Thicken{2}{B}$, we set $s_1 = \sum s(R)$, $d_1 = \sum d(R)$, $e_1 = \sum e(R)$, $t_1 = \sum t(R)$, and refer to $s_1, d_1, e_1, t_1$ as the \emph{satellite sums} of $\Thicken{2}{B}$.
Furthermore, let $c_1$ be the number of compounds of $\Thicken{1}{B}$ and $r_2$ the number of rooms of $\Thicken{2}{B}$.
Since $s(R) + 2d(R) + 2e(R) + 3 t(R)$ is three times the number of small compounds in $R$ plus six for the big compound, the satellite sums satisfy a linear relation, which we state together with a property of short edges connecting satellites in the interior:
\medskip \begin{description}
  \item[(iv)] if $c_1 > 1$, then the satellite sums of $\Thicken{2}{B}$ satisfy $s_1 + 2d_1 + 2e_1 + 3t_1 = 3c_1+3r_2$;
  \item[(v)] any unit length edge connecting blocks of $\Thicken{1}{B}$ inside a room of $\Thicken{2}{B}$ with each other or to satellites in the interior of this room is contained in the interior of this room.
\end{description} \medskip
By construction, there are $s(R) + d(R)$ satellites that are vertices of $R$.
We prove a stronger lower bound on the number of vertices, which also strengthens Claim~{\sf (iii)}.
\begin{lemma}  \label{lem:satellites_on_the_frontier} Assume $r_2 \geq 2$ and let $R$ be a room of $\Thicken{2}{B}$.
  Then the frontier of $R$ has at least
  $6 + \frac{2}{3}s(R) + \frac{4}{3} d(R)$ vertices.
\end{lemma}
\begin{proof}
  Let $p$, $s$, $d$ be the number of non-satellite lattice points, single satellites, double satellites, and write $\perimeter{R}$ for the \emph{perimeter}, which is the length \emph{of} or the number of (short) edges \emph{in} the frontier of $R$.
  To begin note that a satellite in the frontier of $R$ is in the boundary of at most one backyard.
  This is because the external angle is $180^\circ$ at a single satellite and $60^\circ$ at a double satellite.
  The internal angle at any vertex of another room is at least $120^\circ$, so there is not enough space for two backyards around a satellite; see the left two panels in Figure~\ref{fig:satellite}.
  This implies that we may assume that the frontier of $R$ is a simple polygon, or a collection of such.
  Indeed, if the polygon touches itself at a vertex, this must be a non-satellite, which we can duplicate, and if the polygon touches itself along a sequence of edges, we can remove these edges and their shared vertices.
  This operation neither changes the number of single and double satellites, nor does it increase the perimeter.
  A room that contains only one compound can have perimeter as small as $12$, but a room with at least two compounds has significantly larger perimeter, certainly larger than $15$.
  For $\perimeter{R} \leq 15$, we thus get only one compound and, by construction, only $6$ single and no double satellites.
  This implies the claimed inequality.
  We therefore assume \eqref{eqn:perimeter0}, aim at proving \eqref{eqn:perimeter}, and note that \eqref{eqn:perimeter2} follows as the convex combination of \eqref{eqn:perimeter0} and \eqref{eqn:perimeter} with coefficients $\frac{1}{3}$ and $\frac{2}{3}$:
  \begin{align}
    \perimeter{R} &\geq 16 ;
      \label{eqn:perimeter0} \\
    \perimeter{R} &\geq 1 + s + 2d ;
      \label{eqn:perimeter} \\
    \perimeter{R} &\geq \tfrac{1}{3} 16 + \tfrac{2}{3} (1+d+2d) = 6 + \tfrac{2}{3} s + \tfrac{4}{3} d .
      \label{eqn:perimeter2}
  \end{align}
  It remains to prove \eqref{eqn:perimeter}.
  Call the endpoints of an edge in the frontier of $R$ \emph{neighbors}.
  Two neighbors cannot both be double satellites, else they would belong to a common compound, which contradicts that the distance between them is at least $\sqrt{3}$.
  Furthermore, if a double satellite neighbors a single satellite, then this is only possible if they are vertices of an equilateral triangle bounding a backyard, as in Figure~\ref{fig:butterfly} on the left.
  For lack of space around this triangle, its third vertex is a non-satellite.
  The contribution of these three vertices to the right-hand side of \eqref{eqn:perimeter} is $2+1+0 = 3$.
  Hence, we can remove the three edges from the left-hand side and the three vertices from the right-hand side of \eqref{eqn:perimeter} without affecting the validity of the inequality.
  As illustrated in Figure~\ref{fig:butterfly} on the left, two such triangles may touch at a non-satellite vertex, but this does not matter and we can remove the edges and vertices of both triangles from \eqref{eqn:perimeter}.
  \begin{figure}[b]
    \centering \vspace{0.0in}
    \resizebox{!}{1.6in}{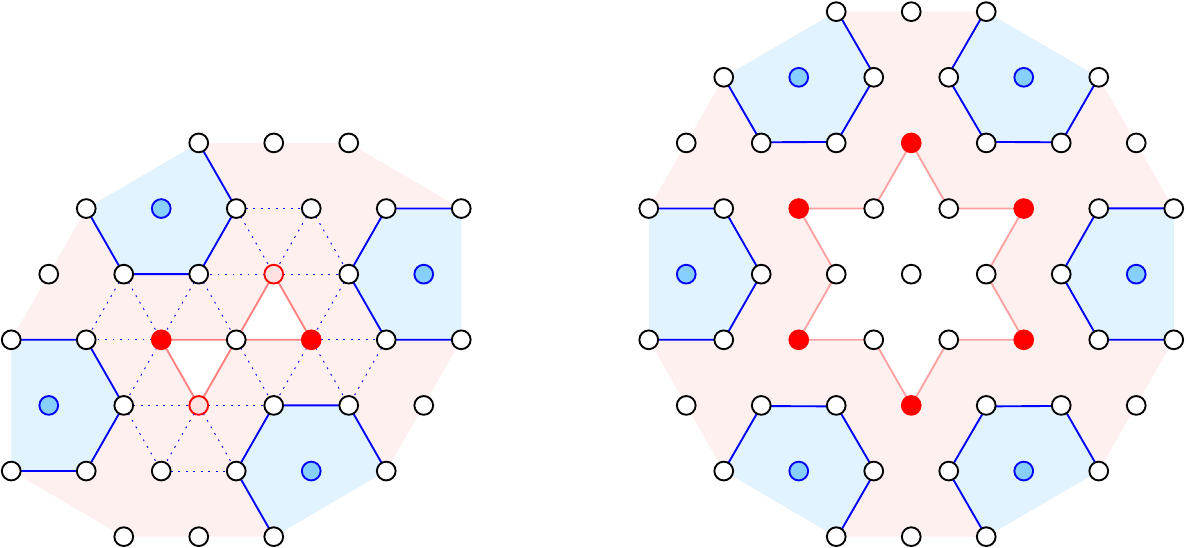}
    \vspace{-0.0in}
    \caption{\emph{Left:} two touching triangular backyards.
    Their shared vertex is a non-satellite, the two \emph{red} vertices are double satellites, and the two \emph{pink} vertices are single satellites.
    \emph{Right:} unique polygon with strictly alternating double satellites and non-satellites.
    On \emph{both sides}, all (partially drawn) \emph{blue} compounds are different and belong to the same (partially drawn) \emph{pink} room.}
    \label{fig:butterfly}
  \end{figure}

  \smallskip
  We can therefore assume that both neighbors of a double satellite are non-satellites.
  Hence, between any two double satellites there is at least one non-satellite, which implies $p \geq d$.
  But $p = d$ only if $p = d = 0$ or there is strict alternation between double satellites and non-satellites.
  It is not possible that all vertices in the frontier are single satellites, because this contradicts that the distance between any two of them is at least $\sqrt{3}$.
  Strict alternation is possible, but only for the polygon of $12$ edges shown in Figure~\ref{fig:butterfly} on the right.
  By assumption, $\Thicken{2}{B}$ has at least two rooms, so not all backyards of $R$ can be bounded by such $12$-gons.
  But this implies $p \geq d+1$, so $\perimeter{R} = p+s+d \geq 1+s+2d$, as claimed.
\end{proof}

To generalize the above concepts to $k \geq 1$, we let $B_{2k-1} \subseteq B_{2k}$ be the vertex sets of two nested components of $\Tree{2k-1}{B}$ and $\Tree{2k}{B}$, and call $\Thicken{k}{B_{2k-1}}$ a \emph{room} and $\Thicken{k}{B_{2k}}$ a \emph{block} of $\Thicken{k}{B}$.
The rooms that share edges join to form \emph{houses}, and the blocks separated by channels that are only $\sqrt{3}/2$ wide join to form \emph{compounds}.
Write $r_k, h_k, b_k, c_k$ for the number of rooms, houses, blocks, compounds of $\Thicken{k}{B}$, $\alpha_k, \beta_k$ for the number of backyards adjacent to at most $2$, at least $3$ houses, and $s_k, d_k, e_k, t_k$ for the satellite sums of $\Thicken{k+1}{B}$.
We can now extend Claims~{\sf (i)} to {\sf (v)} and Lemmas~\ref{lem:number_of_backyards} and \ref{lem:satellites_on_the_frontier} merely by substituting $\Thicken{k}{B}$ for $\Thicken{1}{B}$, $\beta_k$ for $\beta_1$, $c_k$ for $c_1$, etc.
In particular, the extension of Claim~{\sf (iv)} to
\begin{align}
  s_k + 2d_k + 2e_k + 3t_k &= 3c_k+3r_{k+1}
  \label{eqn:iv-extension}
\end{align}
assuming $c_k > 1$ will be needed shortly.
We note that \eqref{eqn:iv-extension} and the extension of Lemma~\ref{lem:satellites_on_the_frontier} can be strengthened, but it is not necessary for the purpose of proving Theorem~\ref{thm:maximum_MST-ratio_of_hexagonal_lattice}.

\subsection{Loop-free Subgraphs}
\label{sec:4.7}

Let $V_k$ be the vertices of $\Thicken{k}{B}$ together with all double and triple satellites that lie in the interior of rooms in $\Thicken{k+1}{B}$, and note that $V_j \cap V_k = \emptyset$ whenever $j \neq k$.
Let $V_k'$ be $V_k$ together with the remaining satellites of $\Thicken{k}{B}$, and note that $V_j \cap V_k' = \emptyset$ if $j < k$, but $V_k'$ and $V_{k+1}$ may share points.
To account for this difference, let $\ell$ be the smallest integer such that $r_{\ell+1} = 1$, and define
\begin{align}
  V &= \left\{ \begin{array}{ll}
          V_1  &  \mbox{\rm if~} \ell = 0; \\
          V_1 \sqcup \ldots \sqcup V_{\ell-1} \sqcup V_{\ell}  &  \mbox{\rm if~} \ell \geq 1 \mbox{\rm ~and~} c_\ell = 1; \\
          V_1 \sqcup \ldots \sqcup V_{\ell-1} \sqcup V_{\ell}'  &  \mbox{\rm if~} \ell \geq 1 \mbox{\rm ~and~} c_\ell > 1.
        \end{array} \right.
     \label{eqn:U}
\end{align}
By construction, all points in $V$ belong to $\BL \setminus B$, and all unit length edges connecting these points are candidates for $\hexMST{\BL \setminus B}$.
We therefore let $U$ be a maximal loop-free graph whose vertices are the points in $V$ and whose edges all have unit length.
Since $U$ has no loops, there is an $\hexMST{\BL \setminus B}$ that contains $U$ as a subgraph.
We are therefore motivated to study the number of edges in $U$.
Using a slight abuse of notation, we denote this number $\card{U}$.
For every $k$, let $U_k$ and $U_k'$ be the subgraphs of $U$ induced by $V_k$ and $V_k'$, respectively.
We first count the edges in $U_1$ and $U_1'$.
\begin{lemma}
  \label{lem:size_of_loop-free_subgraph_I}
  Let $r_1 \geq h_1 \geq b_1 \geq c_1$ be the number of rooms, houses, blocks, and compounds of $\Thicken{1}{B}$, and $s_1, d_1, e_1, t_1$ the satellite sums of $\Thicken{2}{B}$.
  Then 
  \begin{align}
    \card{U_1}  &\geq 2r_1 + h_1 + 3b_1 + (e_1 + t_1) - r_2 - 4 ; \\
    \card{U_1'} &\geq 2r_1 + h_1 + 3b_1 +(s_1 + d_1 + e_1 + t_1) - 5 ,
  \end{align}
  in which we assume $c_1 > r_2 = 1$ for the second inequality.
\end{lemma}
\begin{proof}
  We argue in three steps: first counting edges in $\Frontier{1}{B}$, second counting edges connecting blocks, and third counting edges connecting the satellites.
  In each case, we count only unit length edges, and we make sure that the edges we count do not form loops.

  \smallskip
  For the first step, it is convenient to count \emph{half-edges}, which are the two sides of an edge.
  These two sides either face two rooms, or one faces a room and the other faces the background.
  For a house, $H$, we make its $r(H)$ rooms accessible from the outside by removing $r(H)-1$ edges shared by adjacent rooms plus $1$ edge shared with the background.
  By {\sf (iii)}, each room was originally faced by at least $6$ half-edges, so we still have at least $4r(H) + 1$ of them left.
  Doing this for each house, we make all $r_1$ rooms accessible from the background, and we have at least $4r_1 + h_1$ half-edges left facing these rooms.

  Observe that the convex hull of a house contains at least six of the (short) edges that bound the house.
  One may have been removed, so we still have at least $5$ half-edges facing the background.
  Keeping in mind that the cycles that bound backyards still need to be opened, we now have at least $4r_1 + h_1 + 5h_1$ half-edges and therefore at least $2r_1+3h_1$ edges.
  If a backyard is adjacent to at most two houses, then it has two consecutive (short) edges that enclose an angle less than $\pi$ and that are both shared with the same house.
  Hence, the complementary angle on the side of the house is larger than $\pi$, which implies that these two edges cannot belong to the convex hull of the house.
  We remove one of them and use the half-edge facing the backyard of the other to compensate for the removed half-edge facing the room.
  Since both edges have not yet been accounted for, we still have at least $2r_1+3h_1$ edges.
  If a backyard is adjacent to three or more houses, we also remove one edge, but this time count one less.
  Recalling that $\beta_1$ is the number of such backyards, we still have at least $2r_1 + 3h_1 - \beta_1 \geq 2r_1 + h_1 +2b_1 -2$ edges, in which we get the right-hand side from Lemma~\ref{lem:number_of_backyards}.

  \smallskip
  For the second step, we connect the $b(R)$ blocks inside a common room of $\Thicken{2}{B}$ with $b(R)-1$ short edges.
  A total of $b_1$ blocks are hierarchically organized in $r_2$ rooms, so we add $b_1 - r_2$ short edges to those counted in the first step.
  Similarly, we add $e_1 + t_1$ short edges that connect the double and triple satellites in the interiors of the rooms to the vertices in the frontier of $\Thicken{1}{B}$.
  Finally, we remove two edges to open the meridian and longitudinal cycles of the graph, if they exist.
  The final count is therefore at least $2r_1 + h_1 + 3b_1 + (e_1 + t_1) - r_2 - 4$, which is the claimed lower bound for $\card{U_1}$.

  \smallskip
  For the third step, we assume $c_1 > r_2 = 1$.
  Since there is only one room, there are no shared satellites between different rooms, and we can connect them to the frontier of $\Thicken{1}{B}$ with $s_1 + d_1$ short edges without creating any loop.
  This implies that the number of edges in ${U_1'}$ is at least $2r_1 + h_1 + 3b_1 + (s_1 + d_1 + e_1 + t_1) - 5$, as claimed.
\end{proof}

The bounds in Lemma~\ref{lem:size_of_loop-free_subgraph_I} generalize to $k > 1$, but there are differences.
Most important is the existence of a loop-free graph for thickness $k-1$.
In particular, we have satellites that affect the structure and size of $U_k$ and $U_k'$.
\begin{lemma}
  \label{lem:size_of_loop-free_subgraph_II}
  Let $r_k \geq h_k \geq b_k \geq c_k$ be the number of rooms, houses, blocks, compounds of $\Thicken{k}{B}$, and $s_k, d_k, e_k, t_k$ the satellite sums of $\Thicken{k+1}{B}$.
  Then for $k \geq 2$, we have
  \begin{align}
    \card{U_k}  &\geq (3r_k + \tfrac{1}{3} s_{k-1} + \tfrac{2}{3} d_{k-1}) + 4h_k + 3b_k + (e_k + t_k) - r_{k+1} - 4 ; \\
    \card{U_k'} &\geq (3r_k + \tfrac{1}{3} s_{k-1} + \tfrac{2}{3} d_{k-1}) + 4h_k + 3b_k + (s_k + d_k + e_k + t_k) - 5 ,
  \end{align}
  in which we assume $c_k > r_{k+1} = 1$ for the second inequality.
\end{lemma}
\begin{proof}
  We argue again in three steps: first counting edges in $\Frontier{k}{B}$, second counting edges connecting blocks, and third counting edges connecting to the satellites.
  Each of these three steps is moderately more involved than the corresponding step in the proof of Lemma~\ref{lem:size_of_loop-free_subgraph_I}, and we emphasize the differences.
  
  \smallskip
  The first step starts the construction with Lemma~\ref{lem:satellites_on_the_frontier}, which implies that the rooms in $\Thicken{k}{B}$ are faced by a total of at least $6r_k + \frac{2}{3}s_{k-1} + \frac{4}{3} d_{k-1}$ half-edges.
  After making all rooms accessible to the background, we still have at least $(4r_k + \frac{2}{3}s_{k-1} + \frac{4}{3} d_{k-1}) + h_k$ half-edges.
  Adding the at least $11$ half-edges per house facing the background, we have at least $(4r_k + \frac{2}{3}s_{k-1} + \frac{4}{3} d_{k-1}) + 12h_k$ half-edges and thus at least $(2r_k + \frac{1}{3} s_{k-1} + \frac{2}{3} d_{k-1}) + 6 h_k$ edges.
  Let $\alpha_k$ and $\beta_k$ be the number of backyards adjacent to at most two and at least three houses, respectively.
  By extension of Lemma~\ref{lem:number_of_backyards}, we have $\beta_k \leq 2h_k - 2b_k + 2$.
  We remove an edge per backyard, which for the first type does not affect the current edge count, while the backyards of the second type reduce the count to $(2r_k + \frac{1}{3} s_{k-1} + \frac{2}{3} d_{k-1}) + 4 h_k + 2b_k - 2$.
  
  \smallskip
  For the second step, we connect the blocks of $\Thicken{k}{B}$ inside a common room of $\Thicken{k+1}{B}$ with $b_k - r_{k+1}$ edges.
  Furthermore, we add $r_k$ edges to connect the blocks of $\Thicken{k-1}{B}$ inside a common room of $\Thicken{k}{B}$---which inductively are already connected to each other---to the frontier of this room, and we add at least $e_k + t_k$ edges connecting to the triple satellites of compounds inside the rooms of $\Thicken{k+1}{B}$.
  After removing two additional edges to break the meridian and longitudinal loops, if they exist, we arrive at a lower bound of at least $(3r_k + \frac{1}{3} s_{k-1} + \frac{2}{3} d_{k-1}) + 4 h_k + 3b_k + (e_k + t_k) - r_{k+1} - 4$ edges in $U_k$.
  
  \smallskip
  For the third step, we assume $c_k > r_{k+1} = 1$, in which case we can add at least $s_k + d_k$ edges connecting to the single and double satellites.
  This implies $\card{U_k'} \geq (3r_k + \frac{1}{3} s_{k-1} + \frac{2}{3} d_{k-1}) + 4 h_k + 3b_k + (s_k + d_k + e_k + t_k) - 5$.
\end{proof}

\subsection{Book-keeping}
\label{sec:4.8}

The goal is to show that the average (Euclidean) length of the long edges in $\hexMST{B}$ and the short edges in $\hexMST{\BL \setminus B}$ is at most $\frac{5}{4}$.
We thus assign a \emph{credit} of $\alpha = \frac{1}{4}$ to every short edge and set the \emph{cost} of a long edge to be its Euclidean length minus $\frac{5}{4}$.
For convenience, we set the value of $\alpha$ to $1$ Euro and convert the costs into Euros; see Table~\ref{tbl:hexcosts}.
{\renewcommand{\arraystretch}{1.6}
\begin{table}[hbt]
  \centering \footnotesize
  \begin{tabular}{c||cc|cc||ccc|ccc}
    hex   & $2$ & $2$ & $3$ & $3$ & $4$ & $4$ & $4$ & $5$ & $5$ & $5$ \\
    $L_2$ & $\sqrt{3}$ & $\sqrt{4}$ & $\sqrt{7}$ & \!\!\!$\sqrt{9}$ & $\sqrt{12}$ & $\sqrt{13}$ & $\sqrt{16}$ & $\sqrt{19}$ & $\sqrt{21}$ & $\sqrt{25}$ \\ \hline
    ${\rm cost}$  & {\bf 1.92} & {\bf 3.00} & {\bf 5.58} & {\bf 7.00} & 8.85 & {\bf 9.42} & {\bf 11.00} & 12.43 & {\bf 13.33} & {\bf 15.00}
  \end{tabular}
  \caption{The Euclidean lengths of the edges with hexagonal lengths $2$ to $5$, and their costs in Euros, each truncated beyond the first two digits after the decimal point.}  
  \label{tbl:hexcosts}
\end{table}}

For the accounting, we need the costs of the last two edges for each hexagonal length.
Letting $w_k, x_k$ and $y_k, z_k$ be the costs of the two longest edges with hexagonal length $2k$ and $2k+1$, respectively, we have
\begin{align}
  w_k  &=  \tfrac{1}{\alpha} \left[ \sqrt{4k^2 - 2k + 1} - \tfrac{5}{4} \right], ~
  x_k   =  \tfrac{1}{\alpha} \left[ 2k - \tfrac{5}{4} \right] , 
    \label{eqn:wandx} \\
  y_k  &=  \tfrac{1}{\alpha} \left[ \sqrt{4k^2+2k+1} - \tfrac{5}{4} \right] , ~
  z_k   =  \tfrac{1}{\alpha} \left[ (2k+1) - \tfrac{5}{4} \right] ;
    \label{eqn:yandz}
\end{align}
see Table~\ref{tbl:hexcosts}, which shows the values of $w_1, x_1, y_1, z_1, w_2, x_2, y_2, z_2$ in boldface.
Listing the edges in sequence, we need bounds for the cost differences between consecutive edges:
\begin{align}
  2  \leq   w_k-z_{k-1} &\leq 2.928\ldots ; ~~~
  1.071\ldots   \leq   x_k-w_k   \leq 2 ; 
    \label{eqn:hexdiff12} \\
  2  \leq ~~y_k-x_k~\,  &\leq 2.583\ldots ; ~~~
  1.414\ldots   \leq \,z_k-y_k\, \leq 2 ,
    \label{eqn:hexdiff34}
\end{align}
which are not difficult to prove using elementary computations.
We use accounting with credits and costs to prove that the average (Euclidean) edge length of the two minimum spanning trees is less than $\frac{5}{4}$.
Note that the hexagonal lattice on the torus is obtained by gluing a regular hexagonal portion of the Euclidean hexagonal lattice along opposite sides.
If we choose $12n^2$ points, then this hexagon has $2n+1$ vertices and therefore $2n$ edges per side.
Taking only every other point---so $3n^2$ of the $12n^2$---we still get an integer number of edges per side.
It follows that the $3n^2$ points are the minority color in a $1 : 3$ coloring of the $12n^2$ points.
\begin{lemma}
  \label{lem:total_Euclidean_length_hex}
  Let $\BL$ be the hexagonal lattice with $12n^2$ points and unit minimum distance on the torus, and $B \subseteq \BL$.
  Then $\Length{\MST{B}} + \Length{\MST{\BL \setminus B}} \leq 15 n^2 - \frac{5}{2}$.
\end{lemma}
\begin{proof}
  By \eqref{eqn:Euclideanhex}, it suffices to prove the inequality for $\hexMST{B}$ and $\hexMST{\BL \setminus B}$.
  For $k \geq 1$, we compare the edges of hexagonal length $2k$ and $2k+1$ in $\hexMST{B}$ with the (short) edges in $U_k$ or possibly in $U_k'$.
  Since $\Tree{2k+1}{B} \setminus \Tree{2k-1}{B}$ is the set of these long edges, we can do this in one step by comparing $\Tree{2\ell+1}{B}$ with $U$, for sufficiently large $\ell$ and $U$ as defined right after the definition of $V$ in \eqref{eqn:U}.
  Recall that $r_k$ is the number of components of $\Tree{2k-1}{B}$ or, equivalently, the number of rooms of $\Thicken{k}{B}$.
  These rooms are organized hierarchically into $h_k$ houses, $b_k$ blocks, and $c_k$ compounds.
  Hence, $r_1 \geq h_1 \geq b_1 \geq c_1 \geq r_2$, etc.
  This implies that there are
  \medskip \begin{itemize}
    \item $r_1-h_1$ edges of hexagonal length $2$ and Euclidean length less than $2$ that connect the rooms pairwise inside the $h_1$ houses;
    \item $h_1-b_1$ edges of hexagonal and Euclidean length $2$ that connect the houses pairwise inside the $b_1$ blocks;
    \item $b_1-c_1$ edges of hexagonal length $3$ and Euclidean length less than $3$ that connect the blocks pairwise inside the $c_1$ compounds;
    \item $c_1-r_2$ edges of hexagonal and Euclidean length $3$ that connect the compounds pairwise inside the $r_2$ rooms of $\Thicken{2}{B}$, etc.
  \end{itemize} \medskip
  The costs for these edges are $w_1$, $x_1$, $y_1$, $z_1$, respectively.
  Setting $z_0 = 0$, and generalizing to $k \geq 1$, we observe that the total cost satisfies
  \begin{align}
    \mbox{\rm cost}
      &\leq  \sum\nolimits_{k \geq 1} \left[ w_k (r_k-h_k) + x_k (h_k-b_k) + y_k (b_k-c_k) + z_k (c_k - r_{k+1}) \right] 
        \label{eqn:hex1} \\
      &=     \sum\nolimits_{k \geq 1} \left[ (w_k-z_{k-1}) r_k + (x_k-w_k) h_k + (y_k-x_k) b_k + (z_k-y_k) c_k \right] 
        \label{eqn:hex2} \\
      &\leq  [2r_1+h_1+3b_1+c_1 - 7] + \sum\nolimits_{k \geq 2} [3r_k + h_k + 3b_k + c_k - 8] .
      \label{eqn:hex3}
  \end{align}
  To see how \eqref{eqn:hex3} derives from \eqref{eqn:hex2}, we first make the sums finite by letting $\ell$ be the smallest integer such that $r_{\ell+1} = 1$.
  Then the last non-zero term in \eqref{eqn:hex1} is $z_\ell (c_\ell-r_{\ell+1})$ and, correspondingly, the last term in \eqref{eqn:hex2} is $z_\ell r_{\ell+1} = z_\ell$, which by \eqref{eqn:yandz} is equal to $8 \ell - 1$.
  But this is the same as the sum of constants in \eqref{eqn:hex3}.
  Furthermore, we note that if $r_k = h_k = b_k = c_k = 1$, for every $k$, then \eqref{eqn:hex2} vanishes because \eqref{eqn:hex1} vanishes, and \eqref{eqn:hex3} vanishes because for any $k$ the corresponding sum of four terms minus the constant vanishes.
  Hence, the difference between \eqref{eqn:hex3} and \eqref{eqn:hex2} vanishes.
  To prove the inequality, we reintroduce the variables, which satisfy $r_1 \geq h_1 \geq \ldots \geq c_\ell$, and look at their coefficients.
  The first is $2-w_1+z_0$, which is positive because $w_1 < 2$ and $z_0 = 0$.
  Indeed, using the inequalities in \eqref{eqn:hexdiff12} and \eqref{eqn:hexdiff34}, we observe that the coefficients alternate between positive and negative.
  For example, $3-w_k+z_{k-1}$ is positive because $w_k - z_{k-1} < 3$, and $1-x_k+w_k$ is negative because $x_k-w_k > 1$.
  This implies that the difference is non-negative, so \eqref{eqn:hex3} follows.

  \smallskip
  The difficult cases are the edges of hexagonal lengths $2$ and $3$.
  We therefore consider the special cases in which all edges in $\hexMST{B}$ have Euclidean length at most $\sqrt{3}, \sqrt{4}, \sqrt{7}, \sqrt{9}$, so $h_1=1, b_1=1, c_1=1, r_2=1$, respectively; see Figure~\ref{fig:fourpartitions}.
  From \eqref{eqn:hex3}, we get
  \begin{align}
    \mbox{\rm cost} &\leq
      \left\{ \begin{array}{ll}
        2r_1 - 2  &  \mbox{\rm if~} r_1 > h_1 = 1; \\
        2r_1 + h_1 - 3  &  \mbox{\rm if~} h_1 > b_1 = 1; \\
        2r_1 + h_1 + 3b_1 - 6  &  \mbox{\rm if~} b_1 > c_1 = 1; \\
        2r_1 + h_1 + 3b_1 + c_1 - 7  &  \mbox{\rm if~} c_1 > r_2 = 1.
      \end{array} \right.
      \label{eqn:cost}
  \end{align}
  The cost needs to be paid from the credit contributed by the (short) edges in $U$, which in these four cases is either $U_1$ or $U_1'$.
  Recall that after the conversion, each short edge contributes one Euro of credit, so Lemma~\ref{lem:size_of_loop-free_subgraph_I} provides lower bounds:
  \begin{align}
    \!\!\mbox{\rm credit} &\geq
      \left\{ \begin{array}{ll}
        2r_1 - 1  &  \mbox{\rm if~} r_1 > h_1 = 1; \\
        2r_1 + h_1 - 2  &  \mbox{\rm if~} h_1 > b_1 = 1; \\
        2r_1 + h_1 + 3b_1 - 5  &  \mbox{\rm if~} b_1 > c_1 = 1; \\
        2r_1 + h_1 + 3b_1 + (s_1 + d_1 + e_1 + t_1) - 5  &  \mbox{\rm if~} c_1 > r_2 = 1.
      \end{array} \right.
      \label{eqn:credit}
  \end{align}
  Comparing \eqref{eqn:credit} with \eqref{eqn:cost}, we get $\mbox{\rm cost} \leq \mbox{\rm credit}$ trivially in the first three cases.
  Using Claim~{\sf (iv)}, we get use $s_1 + d_1 + e_1 + t_1 \geq c_1 \geq (s_1+2d_1+2e_1+3t_1) - r_2 = c_1$, which supports the same in the fourth case.
  To compare the cost with the credit in the remaining cases, we use Lemmas~\ref{lem:size_of_loop-free_subgraph_I} and \ref{lem:size_of_loop-free_subgraph_II} to compute a lower bound for the latter, assuming that $\ell > 1$ is the smallest integer for which $r_{\ell+1} = 1$:
  \begin{align}
    \mbox{\rm credit} 
      &\geq  \card{U_1} + \sum\nolimits_{k=2}^{\ell-1} \card{U_k} + \card{U_\ell'} \\
      &\geq  \left[2r_1+h_1+3b_1+(\tfrac{1}{3}s_1+\tfrac{2}{3}d_1+e_1+t_1)-r_2-4 \right] \nonumber \\
      &~~+ \sum\nolimits_{k=2}^{\ell-1} \left[ 3r_k + 4h_k + 3b_k + (\tfrac{1}{3}s_k+\tfrac{2}{3}d_k+e_k+t_k) - r_{k+1} - 4 \right] \nonumber \\
      &~~+  \left[ 3r_\ell + 4h_\ell + 3b_\ell + (s_\ell+d_\ell+e_\ell+t_\ell) - 5 \right] ,
        \label{eqn:credit2}
  \end{align}
  in which we group the terms with index $k-1$ that appear in the bounds for $\card{U_k}$ and $\card{U_k'}$ with the terms that have the same index.
  Using the extension of Claim~{\sf (iv)} to $k \geq 1$ stated in \eqref{eqn:iv-extension}, we get $\frac{1}{3}s_k + \frac{2}{3}d_k+e_k+t_k \geq \frac{1}{3} (s_k + 2d_k + 2e_k + 3t_k) = c_k+r_{k+1}$, so the lower bound in \eqref{eqn:credit2} exceeds the upper bound in \eqref{eqn:hex3}.
  Hence, $\mbox{\rm cost} \leq \mbox{\rm credit}$.
  In other words, the average Euclidean length of the edges in $\hexMST{B}$ and $\hexMST{\BL \setminus B}$ is at most $\frac{5}{4}$.
  It follows that their total Euclidean length is at most $\frac{5}{4} (n^2 - 2)$, which by \eqref{eqn:Euclideanhex} implies the same for $\MST{B}$ and $\MST{\BL \setminus B}$.
\end{proof}

By Lemma~\ref{lem:total_Euclidean_length_hex}, the average Euclidean length of the edges in $\MST{B}$ and $\MST{\BL \setminus B}$ is less than $\frac{5}{4}$.
Together with \eqref{eqn:minMSTratiohex}, this implies Theorem~\ref{thm:maximum_MST-ratio_of_hexagonal_lattice}.

\section{Discussion}
\label{sec:5}

This paper proves bounds on the supremum and infimum of the maximum MST-ratio for finite sets, as well as of the supremum MST-ratio for lattices in the plane.
There are many directions of generalization, and their connection to the topological analysis of colored point sets started in \cite{CDES23} provides a potential path to relevance outside of mathematics.
\medskip \begin{itemize}
  \item What about sets in the plane that are less restrictive than lattices but still disallow arbitrarily dense clusters of points, such as periodic sets or Delone sets?
  A first result in this direction is the lower bound of $1 + 1 / (11(2 c + 1)^2)$ for the maximum MST-ratio of a set of $n$ points with spread at most $c \sqrt{n}$ proved in \cite{DPT23}.
  \item What about partitions of $\AF \subseteq \Rspace^2$ into three or more sets?
  For example, is it true that the supremum MST-ratio of the hexagonal lattice partitioned into three subsets is $\sqrt{3}$, as realized by the unique partition into three congruent hexagonal grids?
  Is $\sqrt{3}$ the infimum, over all lattices in $\Rspace^2$, of the supremum, over all partitions into three subsets?
  \item What about three and higher dimensions?
  Consider for example the $\FCC$ lattice in $\Rspace^3$ (all integer points whose sums of coordinates are even), and partition it into $2\FCC$ and the rest.
  The MST-ratio of this example is $\frac{9}{8} = 1.125$.
  Is it true that this is the supremum MST-ratio of the $\FCC$ lattice?
  Is $1.125$ the infimum, over all lattices in $\Rspace^3$, of the supremum, over all partitions into two subsets?
\end{itemize} \medskip
Beyond these extensions in discrete geometry, it would be interesting to study the MST-ratio stochastically, to determine the computational complexity of the maximum MST-ratio, and to frame notions of mingling as measured by homology classes of dimension $1$ and higher in elementary geometric terms.
\Skip{
  Recalling that the MST-ratio relates to mingling as measured with $0$-dimensional homology, it would also be interesting to understand the measures with $1$- and higher-dimensional homology, as described in \cite{CDES23}, if at all possible in elementary geometric terms.
  Importantly, stochastic versions of all the above questions are worth studying, both theoretically and experimentally?
  For the latter undertaking, it would be useful to have a fast algorithm that computes the maximum MST-ratio of a finite set of points.
  Is there a polynomial-time algorithm, and if not, to what extent can the maximum MST-ratio be approximated in polynomial time?
} 

\subsubsection*{Acknowledgments}

{\footnotesize
Part of the research described in this paper was conducted during the sabbatical of the third author at the University of Kyoto in Japan.  
We thank the Institute of Advanced Study in Kyoto for the hospitality afforded to the authors of this paper.}


\appendix

\section{Connection to Chromatic Persistence}
\label{app:A}

As mentioned in the introduction, the study of the MST-ratio is motivated by a recent topological data analysis method for measuring the ``mingling'' of points in a colored configuration;
see Figure~\ref{fig:six-pack_hexgrid_10x10}, which shows six persistence diagrams measuring various aspects of the mingling in a bi-colored configuration.
This appendix addresses the meaning of some of these diagrams and explains the connection to the MST-ratio, while referring to \cite{CDES23} for a detailed account of the method.
In particular, we short-cut the description by ignoring the discrete structures that are necessary for the algorithm.
We first sketch the general background from \cite{EdHa10} and \cite{CEHM09}, and then explain the specific setting that motivates the MST-ratio.

\smallskip
Let $\AF \subseteq \Rspace^2$ be a finite set of points, $\chi \colon \AF \to \{0,1\}$ a bi-coloring, and write $B = \chi^{-1} (0)$ and $C = A \setminus B = \chi^{-1} (1)$.
Let $\aaa \colon \Rspace^2 \to \Rspace$ be the function that maps every $x \in \Rspace^2$ to the minimum Euclidean distance between $x$ and the points in $\AF$, and let $\bbb \colon \Rspace^2 \to \Rspace$ and $\ccc \colon \Rspace^2 \to \Rspace$ be the similarly defined functions for $B$ and $C$.
Furthermore, write $\AF_r = \aaa^{-1} [0,r]$, $B_r = \bbb^{-1} [0,r]$, and $C_r = \ccc^{-1} [0,r]$ for the sublevel sets at distance threshold $r \geq 0$.
Each is a union of disks with radius $r$ centered at the points of $\AF$, $B$, and $C$, respectively.
The inclusions $B_r \subseteq \AF_r$ and $C_r \subseteq \AF_r$ induce homomorphisms in $p$-th homology, $b_r \colon \Hgroup{p}{B_r} \to \Hgroup{p}{\AF_r}$ and $c_r \colon \Hgroup{p}{C_r} \to \Hgroup{p}{\AF_r}$, for each dimension $p \in \Zspace$ and every threshold $r \geq 0$.
Assuming field coefficients in the construction of the homology groups, the latter are vector spaces and the homomorphisms are linear maps.

We also have $\AF_r \subseteq \AF_s$ whenever $r \leq s$, so there are also linear maps from $\Hgroup{p}{\AF_r}$ to $\Hgroup{p}{\AF_s}$.
By now it is tradition in the field to consider the \emph{filtration} of the $\AF_r$, for $r$ from $0$ to $\infty$, and the corresponding sequence of homology groups together with the linear maps between them.
Reading this sequence from left to right, we see homology classes being born and dying.
There is a unique way to pair the births with the deaths that regards the identity of the classes, and the \emph{persistence diagram} summarizes this information by drawing a point $(r,s) \in \Rspace^2$ for every homology class that is born at $\AF_r$ and dies entering $\AF_s$; see e.g.\ \cite[Chapter VII]{EdHa10}.
Every death is paired with a birth, but it is possible that a birth remains unpaired---when the homology class is of the domain---in which case the corresponding point is at infinity.
We write $\Dgm{p}{\aaa}$ for the persistence diagram defined by the sublevel sets of $\aaa$, noting that it is a multi-set of points vertically above the diagonal.

\smallskip
Besides $\Dgm{p}{\aaa}$, we consider $\Dgm{p}{\bbb}$ and $\Dgm{p}{\ccc}$, which are the persistence diagrams of the sublevel sets of $\bbb$ and $\ccc$, respectively, and work with the disjoint union, $B_r \sqcup C_r$.
Conveniently, the $p$-th persistence diagram of $\bbb \sqcup \ccc \colon \Rspace^2 \sqcup \Rspace^2 \to \Rspace$ is the disjoint union of $\Dgm{p}{\bbb}$ and $\Dgm{p}{\ccc}$, for all $p$.
Write $b_r \oplus c_r \colon \Hgroup{p}{B_r} \oplus \Hgroup{p}{C_r} \to \Hgroup{p}{\AF_r}$ for the corresponding map in homology.
As proved in \cite{CEHM09}, the sequence of images of the $b_r \oplus c_r$ admit linear maps between them and thus define another persistence diagram, denoted $\Dgm{p}{{\rm im\;} \bbb \sqcup \ccc \to \aaa}$.
Similarly, the kernels of the $b_r \oplus c_r$ define a persistence diagram, denoted $\Dgm{p}{{\rm ker\;} \bbb \sqcup \ccc \to \aaa}$.
To simplify the notation, we write $\kappa_r = b_r \oplus c_r$ and use mnemonic notation to indicate whether a persistence diagram belongs to the domain, image, or kernel of the map:
\begin{align}
  \Domain{p}{\kappa} &= \Dgm{p}{\bbb \sqcup \ccc} , \\
  \Image{p}{\kappa} &= \Dgm{p}{{\rm im\;} \bbb \sqcup \ccc \to \aaa} , \\
  \Kernel{p}{\kappa} &= \Dgm{p}{{\rm ker\;} \bbb \sqcup \ccc \to \aaa}.
\end{align}
The \emph{$1$-norm} of a persistence diagram, $D$, is the sum of the absolute differences between birth- and death-coordinates over all points in $D$, denoted $\norm{D}_1$.
To cope with points at infinity, we use a cut-off---e.g.\ the maximum finite homological critical value, denoted $\cutoff$---so that the contribution of a point at infinity to the $1$-norm is finite.
\begin{figure}[hbt]
    \centering
    \includegraphics[width=
\textwidth]{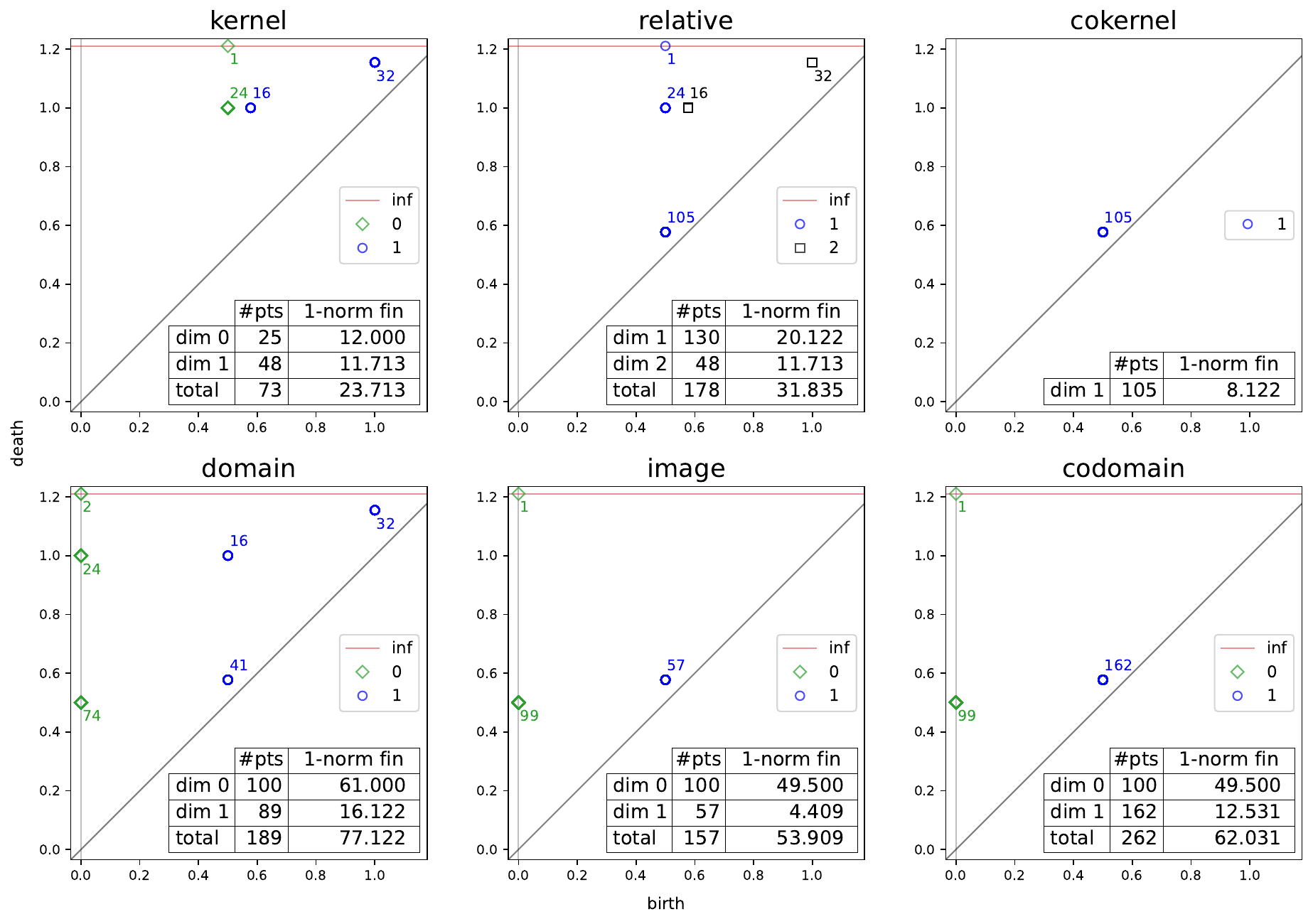} \\
    \vspace{-0.1in}
    \caption{\footnotesize{The six-pack for the $10 \times 10$ portion of the hexagonal lattice with coloring as in Figure~\ref{fig:rhombus}.
    Important for the current discussion are the \emph{diamond-shaped} points in the domain, image, and kernel diagrams.
    To get the MST-ratio, the $1$-norms of the diagrams are computed while ignoring the points at infinity, giving giving $61.0$ and $49.5$ for the domain and the image diagrams, respectively.}
    Compare the ratio of $1.232\ldots$ with the upper bound of $1.25$ proved in Theorem~\ref{thm:maximum_MST-ratio_of_hexagonal_lattice}.}
    \label{fig:six-pack_hexgrid_10x10}
\end{figure}

\smallskip
The kernel, domain, and image form a short exact sequence that splits, which implies $\norm{\Image{p}{\kappa}}_1 + \norm{\Kernel{p}{\kappa}}_1 = \norm{\Domain{p}{\kappa}}_1$; see \cite[Theorem~5.3]{CDES23}.
For dimension $p = 0$, all three $1$-norms can be rewritten in terms of minimum spanning trees.
Indeed, $\norm{\Dgm{0}{\bbb}}_1 = \frac{1}{2} \Length{\MST{B}} + \cutoff$ because every edge in the minimum spanning tree of $B$ marks the death of a connected component in the sublevel set, and $\cutoff$ is contributed by the one component that never dies.
Similarly, $\norm{\Dgm{0}{\ccc}}_1 = \frac{1}{2} \Length{\MST{C}} + \cutoff$, which implies \eqref{eqn:domainnorm}:
\begin{align}
  \norm{\Domain{0}{\kappa}}_1 
    &= \norm{\Dgm{0}{\bbb}}_1 + \norm{\Dgm{0}{\ccc}}_1
     = \tfrac{1}{2} \Length{\MST{B}} + \tfrac{1}{2} \Length{\MST{C}} + 2\cutoff; 
     \label{eqn:domainnorm} \\
  \norm{\Image{0}{\kappa}}_1 &= \tfrac{1}{2} \Length{\MST{\AF}} + \cutoff.
     \label{eqn:imagenorm}
\end{align}
Since persistence diagrams are stable, as originally proved in \cite{CEH07}, these relations imply that minimum spanning trees are similarly stable.
\eqref{eqn:imagenorm} deserves a proof.
There are two ways a connected component of $B_r$ can die in the image: by merging with a component of $C_r$ or with another component of $B_r$.
In the first case, the death corresponds to an edge of $\MST{\AF}$ that connects a point in $B$ with a point in $C$, and in the second case, it corresponds to an edge of $\MST{\AF}$ that connects two points in $B$.
There is also the symmetric case in which the edge connects two points in $C$.
This establishes a bijection between the deaths in $\Image{0}{\kappa}$ and the edges of $\MST{\AF}$.
There is one component that never dies, which accounts for the extra cut-off term and implies \eqref{eqn:imagenorm}.

\smallskip
The $1$-norm of the kernel diagram is the difference between the $1$-norms of the domain diagram and the image diagram: $\norm{\Kernel{0}{\kappa}}_1 = \norm{\Domain{0}{\kappa}}_1 - \norm{\Image{0}{\kappa}}_1$.
It thus makes sense to call $\norm{\Image{0}{\kappa}}_1 / \norm{\Domain{0}{\kappa}}_1$ and $\norm{\Kernel{0}{\kappa}}_1 / \norm{\Domain{0}{\kappa}}_1$ the \emph{image share} and \emph{kernel share}, respectively.
Observe that both are real numbers between $0$ and $1$ and that they add up to $1$.
The intuition is that the kernel share is a measure of the amount of ``$0$-dimensional mingling'' of $B$ and $C$.
In other words, the smaller the image share, the more the two colors mingle.
We therefore get
\begin{align}
  \MSTratio{\AF}{B} &= \frac{\Length{\MST{B}} + \Length{\MST{C}}}
                            {\Length{\MST{A}}}
                     = \frac{\norm{\Domain{0}{\kappa}}_1 - 2\cutoff}
                            {\norm{\Image{0}{\kappa}}_1 - \cutoff} ,
    \label{eqn:imageshare}
\end{align}
for the MST-ratio, which besides the cut-off terms is the reciprocal of the image share.
Hence, the larger the MST-ratio the more the two colors mingle.
In this interpretation, Theorem~\ref{thm:maximum_MST-ratio_for_lattices} says that among all lattices in $\Rspace^2$, the hexagonal lattice is most restrictive to mingling as it does not permit MST-ratios larger than the inf-max, which for $2$-dimensional lattices is $1.25$.

\end{document}
\clearpage
\section{Notation}
\label{app:N}

\begin{table}[h!]
  \centering
  \begin{tabular}{ll}
    $\AF \subseteq \Rspace^2; \AL \subseteq \Rspace^2$
      &  finite set; lattice \\
    $B \subseteq \AF$
      &  partition into two \\
    $\Length{\MST{\AF}}$
      &  length of minimum spanning tree \\
    $\MSTratio{\AF}{B}$
      &  MST-ratio \\
    $\InfSup, \SupSup$
      &  inf-sup, sup-sup for lattices \\
    \\
    $1 = \norm{\uuu} \leq \norm{\vvv} = \nu$
      &  vectors spanning a lattice, and their lengths \\
    $n_r, m_r, p_r, q_r$
      &  numbers of points; columns inside $[-r,r]^2$ \\
    $\alpha, \beta$
      &  integer coefficients \\
    $T(B_n); \ell_i$
      &  minimum spanning tree; lengths of edges \\
    $B_n \subseteq \AL_n; B_n' \subseteq \AL_n'$
      &  (portion of) lattice in plane and torus \\
    $T, T', T''; E, E', E''$
      &  trees; sets of edges \\
    $a, b, c; x, y, z \subseteq V''$
      &  vertices \\
    \\
    $\BL = \BL_1 \sqcup \BL_3$
      &  $1:3$ partition of hexagonal lattice \\
    $\Edist{u}{v}, \hexdist{u}{v}$
      &  Euclidean, hexagonal distance \\
    $\Hdisk, k\Hdisk$
      &  unit disk of hexagonal norm, $k$-th thickening \\
    $\Thicken{k}{B}, \Frontier{k}{B}$
      &  thickening, frontier \\
    $\hexMST{B}, \Tree{k}{B}$
      &  hex norm MST, subset of edges \\
    $\hexMST{\BL \setminus B}, \Uree{k}{B}, \Vree{k}{B}$
      &  hex norm MST, loop-free graphs \\
    $B_1 \subseteq B_2' \subseteq B_2 \subseteq B_3'$
      &  vertex sets of components \\
    $r_k \geq h_k \geq b_k \geq c_k$
      &  \#rooms, \#houses, \#blocks, \#compounds \\
    $s(R), d(R), e(R), t(R)$
      &  \#single, double, double, triple satellites of room \\
    $s_k, d_k, e_k, t_k$
      &  satellite sums of $\Thicken{k+1}{B}$ \\
    $w_k, x_k, y_k, z_k$
      &  costs of edges \\
    \\
    $\aaa, \bbb, \ccc \colon \Rspace^2 \to \Rspace$
      &  min distance maps on the plane \\
    $\AF_r, B_r, C_r$
      &  sublevel sets \\
    $\Hgroup{p}{\AF_r}, \Hgroup{p}{B_r}, \Hgroup{p}{C_r}$
      &  (reduced) homology groups \\
    $b_r, c_r, h_{r,s}$
      &  linear maps between homology groups \\
    $\Dgm{p}{\aaa}, \Dgm{p}{\bbb}, \Dgm{p}{\ccc}$
      &  persistence diagrams \\
    $\norm{D}_1$
      &  $1$-norm of persistence diagram \\
    $\bbb \sqcup \ccc; \kappa_r = b_r \oplus c_r$
      &  disjoint union; direct sum \\
    $\Domain{p}{\kappa}, \Image{p}{\kappa}, \Kernel{p}{\kappa}$
      &  domain, image, kernel diagrams
  \end{tabular}
  \caption{Notation used in the paper.}
  \label{tbl:Notation}
\end{table}


\begin{thebibliography}{21}

{\footnotesize


\bibitem{Bor26}
{\sc O.\ Bor\r{u}vka.}
O jist\'{e}m probl\'{e}mu minim\'{a}ln\'{i}m.
\emph{Pr\'{a}ce Mor.\ P\v{r}\'{i}rodov\v{e}d Spol.\ v Brn\v{e} (Acta Societ.\ Scient.\ Natur Moravicae)} {\bf 3} (1926), 37--58.

\bibitem{Car09}
{\sc G.\ Carlsson.}
Topology and data.
\emph{Bull.\ Amer.\ Math.\ Soc.} {\bf 46} (2009), 255--308.


\bibitem{ChGr85}
{\sc F.R.K.\ Chung and R.L.\ Graham.}
A new bound for Euclidean Steiner minimal trees.
In \emph{Discrete Geometry and Convexity}, Annals N.Y.\ Acad.\ Sci.\ {\bf 440}, New York, 1985, 328--346.

\bibitem{CEH07}
{\sc D.\ Cohen-Steiner, H.\ Edelsbrunner and J.\ Harer.}
Stability of persistence diagrams.
\emph{Discrete Comput.\ Geom.} {\bf 37} (2007), 103--120.

\bibitem{CEHM09}
{\sc D.\ Cohen-Steiner, H.\ Edelsbrunner, J.\ Harer and D.\ Morozov.}
Persistent homology for kernels, images, and cokernels.
In ``Proc.\ 20th Ann.\ ACM-SIAM Sympos.\ Discrete Alg., 2009'', 1011--1020.

\bibitem{CDES23}
{\sc S.\ Cultrera di Montesano, O.\ Draganov, H.\ Edelsbrunner and M.\ Saghafian.}
Chromatic alpha complexes.
\texttt{arXiv:2212.03128 [math.AT]}, 2023.

\bibitem{DPT23}
{\sc A.\ Dumitrescu, J.\ Pach and G.\ T\'{o}th.}
Two trees are better than one.
\texttt{arXiv:2312.09916 [math.CO]}, 2023.

\bibitem{EdHa10}
{\sc H.\ Edelsbrunner and J.L.\ Harer.}
\emph{Computational Topology. An Introduction.}
Amer.\ Math.\ Soc., Providence, Rhode Island, 2010.


\bibitem{GiPo68}
{\sc E.N.\ Gilbert and H.O.\ Pollack.}
Steiner minimal trees.
\emph{SIAM J.\ Appl.\ Math.} {\bf 16} (1968), 1--29.


\bibitem{JaKo34}
{\sc V.\ Jarn\'{i}k and M.\ K\"{o}ssler.}
O minim\'{a}ln\'{i}ch grafech, obsahuj\'{i}c\'{i}ch $n$ dan\'{y}ch bod\r{u}.
\emph{\v{C}asopis pro P\v{e}stov\'{a}ni Matematiky a Fysiky} {\bf 63} (1934), 223--235.

\bibitem{Kru56}
{\sc J.B.\ Kruskal.}
On the shortest spanning tree of a graph and the traveling salesman problem.
\emph{Prof.\ Amer.\ Math.\ Soc.} {\bf 7} (1956), 48--50.

\bibitem{Zhi15}
{\sc B.\ Zhilinskii.}
\emph{Introduction to Louis Michel's Lattice Geometry through Group Action.}
EDP Sciences, ENRS Editions, Paris, France, 2015.

}
\end{thebibliography}
\end{document}